\documentclass[10pt]{llncs}
\pdfoutput=1

\usepackage{graphicx}
\usepackage{url}

\newif \ifFullversion \Fullversionfalse   
\newif \ifTitlepage \Titlepagefalse   

\usepackage{times}
\usepackage{mathptm}
\usepackage{cite}
\usepackage{graphicx}
\usepackage{color}
\usepackage{cramped}
\usepackage{subfigure}

\DeclareSymbolFont{AMSb}{U}{msb}{m}{n}
\DeclareSymbolFontAlphabet{\Bbb}{AMSb}
\def\R{\ensuremath{\Bbb R}}

\makeatletter
\def\hb@xt@{\hbox to }
\makeatother

\let\oldendproof\endproof
\def\endproof{\qed\oldendproof}

\newcommand{\comment}[1]{\relax}

\begin{document}

\ifTitlepage
   \begin{titlepage}
   \thispagestyle{empty}
   \setcounter{page}{0}
\fi

\title{Straight Skeletons of Three-Dimensional Polyhedra}

\author{Gill Barequet\inst{1} \and David Eppstein\inst{2} 
	\and Michael T. Goodrich\inst{2} \and Amir Vaxman\inst{1}}

\institute{Dept. of Computer Science\\
Technion---Israel Institute of Technology\\
\email{\{barequet,avaxman\}(at)cs.technion.ac.il}\\[0.1in]
\and
Computer Science Department\\
University of California, Irvine\\
\email{\{eppstein,goodrich\}(at)ics.uci.edu}}

\maketitle

\begin{abstract}
\ifTitlepage\large\fi
This paper studies the straight skeleton of polyhedra in three dimensions.
We first address voxel-based polyhedra (polycubes),
formed as the union of a collection of cubical (axis-aligned) voxels.
We analyze the ways in which the skeleton may intersect each voxel of
the polyhedron, and show that the skeleton may be constructed by a
simple voxel-sweeping algorithm taking constant time per voxel.
In addition, we describe a more complex algorithm for straight skeletons
of voxel-based polyhedra, which takes time proportional to the area of the
surfaces of the straight skeleton rather than the volume of the polyhedron.
We also consider more general polyhedra with axis-parallel edges and faces,
and show that any $n$-vertex polyhedron of this type has a straight
skeleton with $O(n^2)$ features. We provide algorithms for
constructing the straight skeleton, with running time
$O(\min(n^2\log n,k\log^{O(1)} n))$ where $k$ is the output complexity.  
Next, we discuss the straight skeleton of a general nonconvex polyhedron.
We show that it has an ambiguity issue, and suggest a consistent method to
resolve it.
We prove that the straight skeleton of a general polyhedron has a
superquadratic complexity in the worst case.
\ifFullversion
Thus, we show that straight skeletons are strictly simpler for orthogonal
polyhedra than they are for more general polyhedra.
\fi
Finally, we report on an implementation of a simple algorithm for the
general case.
\end{abstract}

\ifTitlepage
   \end{titlepage}
\fi

\section{Introduction}

The straight skeleton is a geometric construction that reduces
two-dimensional shapes---polygons---to one-dimensional sets of line
segments approximating the same shape. 
It is defined in terms of an offset process in which edges 
move inward, remaining straight and meeting at vertices. When a vertex meets an offset edge, 
the process continues within the two pieces so formed.
The straight line segments traced out
by vertices during this offset process define 
the straight skeleton.
Introduced in
1995 by Aichholzer {\it et al.}~\cite{AicAurAlb-JUCS-95,AicAur-COCOON-96}, 
the two-dimensional straight skeleton
has since found many applications, including surface
folding~\cite{DemDemLub-JCDCG-98}, offset curve
construction~\cite{EppEri-DCG-99}, interpolation of three-dimensional
surfaces from cross-section contours~\cite{1046646},
automated interpretation of geographic data~\cite{HauSes-GMR-04},
polygon decomposition~\cite{TanVel-SoCG-03}, and graph
drawing~\cite{BagRaz-CI-04}. 
Compared to other well-known types of skeleton, 
the straight skeleton is more complex to
compute~\cite{EppEri-DCG-99,CheVig-SODA-02}, but its simple geometric
form, comprised exclusively of line segments, offers advantages in applications.
The best known alternative, the medial axis~\cite{b-tends-67}, consists of 
both linear and quadratic curve segments.
Thus, of the two, only the straight skeleton characterizes the shape of a  polygon while preserving its linear nature.

It is natural, then, to try to extend algorithms for straight skeleton
construction to three dimensions. In three dimensions, a skeleton is a
two-dimensional approximation of a three-dimensional shape such as a
polyhedron. The most well-known type of three-dimensional skeleton, the
medial axis, has found applications, for instance, in mesh
generation~\cite{PriArmSab-IJNME-95} and surface
reconstruction~\cite{BitTsiGas-CGF-95}.
Unlike its two-dimensional counterpart,
the 3D medial axis can be quite complex, both combinatorially
and geometrically. Thus, we would like an alternative way to characterize 
the shape of
three-dimensional polyhedra using a simpler type of two-dimensional skeleton.

\ifFullversion
Of particular interest is a mechanism for characterizing
the shape of three-dimensional data, especially data that are generated  
from digital sources, as would be come from CT, MRI, ultrasound, 
seismic imaging,
and similar sources, to be used, for example, to study 
the shapes of molecules, body parts, or oil fields.
In these cases, the data is specified as a union of (axis-aligned)
voxels or as orthogonal polyhedra.
\fi

\subsection{Related Prior Work}

Despite the large amount of work on 2D straight skeletons cited
above, we are not aware of any prior work on 3D straight skeletons, 
other than Demaine {\it et al.}~\cite{HingedPolyforms3D_WADS2005}, who
mention the existence and basic properties of 3D straight skeletons,
but do not study them in any detail with respect to
their algorithmic, combinatorial, or geometric properties.

Held~\cite{Held94} showed that in the worst case, the complexity of the
medial axis of a convex polyhedron of complexity $n$ is $\Omega(n^2)$,
which implies a similar bound for the 3D straight skeleton.
Perhaps the most relevant prior work is on shape characterization using
the 3D medial axis.
This structure is defined from a 3D polyhedron by
considering each face, edge, and vertex
as being a distinct object and then constructing the
3D Voronoi diagram of this set of objects.
Thus, the medial axis is the loci of points in $\R^3$ that are
equidistant to at least two objects. The best
known upper bound for its combinatorial complexity is
$O(n^{3+\epsilon})$~\cite{s-atubl-94}, 
for any fixed constant $\epsilon>0$, and even for the
special case of lines in space it is a well-known open problem in
computational geometry whether the Voronoi diagram (a space subdivision
having the medial axis as its boundary) has subcubic combinatorial
complexity~\cite{CheKedSha-Algs-98,KolSha-SJC-03}.\footnote{See also \url{http://maven.smith.edu/~orourke/TOPP/P3.html}.}
Additionally, the medial axis consists of intersecting 
pieces of planes and conic surfaces, presenting 
significant complications to algorithms that attempt to construct 3D
medial axes.

Because of these drawbacks, a number of researchers have studied
algorithms for computing approximate 3D medial axes.
Sherbrooke {\it et al.}~\cite{spb-amat3-96} take a numerical approach,
giving an algorithm that traces out the curved edges of the 3D medial skeleton.
Culver {\it et al.}~\cite{ckm-acmap-99} also design a curve-tracing
algorithm, but they use exact arithmetic to
compute an exact representation of a 3D medial axis.
In both cases, the running time depends on both the combinatorial and
geometric complexity of the medial axis.
Foskey {\it et al.}~\cite{flm-ecsma-03} study an approximation based
on relaxed distance calculations.
In particular, they construct an approximate medial axis using a
voxel-based approach that runs in time $O(nV)$, where $n$ is the
number of features of the input polyhedron
and $V$ is the volume of the voxel mesh that contains it.
Sheehy {\it et al.}~\cite{sar-sdmsc-96} instead take the approach of
using the 3D Delaunay triangulation of a cloud of points on the
surface of the input polyhedron to compute and approximate 3D medial
axis.
Likewise, Dey and Zhao~\cite{dz-amavc-04} study the 3D medial axis as a
subcomplex of the Voronoi diagram
of a sampling of points approximating the input 
polyhedron.

\subsection{Our Results}

In this paper we provide the following results.

\begin{itemize}
\item We study the straight skeleton of orthogonal polyhedra formed
      as unions of cubical voxels. We analyze the ways in
      which the skeleton may intersect each voxel of the polyhedron, and
      show that the skeleton may be constructed by a simple voxel sweeping
      algorithm taking constant time per voxel.

\item We describe a more complex algorithm for straight skeletons of
      voxel-based polyhedra, which, rather than taking time proportional to
      the volume of the polyhedron takes time proportional to the area of the
      straight skeleton or, equivalently, the number of voxels it intersects.

\item We consider more general polyhedra with axis-parallel edges and
      faces, and show that any $n$-vertex polyhedron of this type has a
      straight skeleton with $O(n^2)$ features. We provide two algorithms
      for constructing the straight skeleton, resulting in a combined 
      running time of $O(\min(n^2\log n,k\log^{O(1)} n))$, where $k$ is the
      output complexity.

\item We discuss the difficulties of unambiguously defining straight
      skeletons for non-axis-aligned polyhedra and suggest a consistent
      method for resolving these ambiguities.
      We show that a general polyhedron, the straight skeleton can,
      in the worst case, have superquadratic complexity.
      Thus, straight skeletons are strictly simpler for orthogonal polyhedra
      than they are for more general polyhedra.  
      We also describe a simple algorithm for computing the straight skeleton
      in the general case.
\end{itemize}

\section{Voxel Polyhedra}

In this section we consider the case in which the polyhedron is a polycube,
that is, a rectilinear polyhedron all of whose vertices have integer
coordinates.  The ``cubes'' making up the polyhedron are also called voxels.

For voxels, and more generally for orthogonal polyhedra, the straight 
skeleton is a superset of the $L_\infty$ Voronoi diagram; the added 
boundaries in the straight skeleton resolve any ambiguities concerning 
which cells of the diagram belong to which features of the input 
polyhedron. Due to this relationship with Voronoi diagrams, the straight 
skeleton is significantly easier to compute for orthogonal inputs than 
in the general case.

As in the general case, the straight skeleton of a polycube can be modeled
by offsetting the boundary of the polycube inward, and tracing the movement
of the boundary. During this sweep, the boundary forms a moving front (or
fronts) whose features are faces, edges, and vertices.  An edge can be
either convex or concave, while a vertex can be convex, concave,
or a saddle.
In the course of this process, features may disappear or appear.

The sweep starts at time~0; at this time the front is the boundary of the
polycube.  In the first time unit we process all the voxels adjacent to the
boundary.  In the $i$th round (for $i \geq 1$) we process all the voxels
adjacent to voxels processed in the $(i-1)$st round, that have never been
processed before.  Processing a voxel means the computation of the piece of
the skeleton lying within (or on the boundary) of the voxel.  During this
process, the polycube is shrunk, and may be broken into several components if
it is not convex.  The process continues for every piece separately until
it vanishes, that is, there are no more voxels to process.

\ifFullversion
We will now show that computing the straight skeleton can be done
efficiently, either in time proportional to the volume (number of voxels) of
the polycube or in an output-sensitive manner.
\fi

\subsection{A Volume Proportional-Time Algorithm}

\begin{theorem}
   The combinatorial complexity of the straight skeleton of a polycube of
   volume $V$ is $O(V)$.  The skeleton can be computed in $O(V)$ time.
\end{theorem}

\begin{proof}
   We prove the two parts of the theorem simultaneously by analyzing the
   skeleton computation procedure, in which the boundary of the polycube is
   swept inward and the movement of its features
\ifFullversion%
   (vertices, edges, and faces)
\fi%
   is traced.  We analyze separately the time needed to find
   the voxels processed in each sweep step
   and the time needed to process each voxel.
\ifFullversion%
   The key observations are
\else%
   We show
\fi%
   that the entire sweep can be performed in time linear in the number of
   voxels, the complexity of the skeleton within every voxel is constant, and
   the portion of the skeleton within every voxel can be computed in constant
   time.

   The sweep starts at time~0 at the boundary of the polycube.
   In the first
   round we process all the voxels adjacent to the boundary.
\ifFullversion%
   These can be found in $O(V)$ time.
\fi%
   In the $i$th round (for $i \geq 1$) we process all
   the voxels adjacent to voxels processed in the $(i-1)$st round, that have
   never been processed before.
\ifFullversion%
   Since the total number of face adjacencies
   of voxels is $\Theta(V)$, the entire sweeping process takes $\Theta(V)$ steps,
   where each step is the processing of a single voxel.

\fi%
   When sweeping the boundary inward during one round of the process, each
   feature of the boundary (vertex, edge, or face) moves inward one
   L$_\infty$ unit.
   For clarity of exposition, we will analyze the process within eighths of voxels instead of full
   voxels.  This will reduce the number of possible cases, since we will have
   to consider all combinations of vertices/edges/facets hit by the moving
   front(s) only on three facets instead of the six facets that a full voxel
   can be hit simultaneously on.  Consider an eighth of the voxel that is about to be
   swept by the moving front(s).  This ``subvoxel'' can be hit in many
   combinations of its corner vertex, the three edges adjacent to this corner,
   and the three faces adjacent to this corner.  Moreover, it may be hit by
   multiple portions of the moving front in a single
   feature of the subvoxel, in two features, one containing the
   other, or multiple features with more complex containment relations.
   Nevertheless, the number of different cases is finite, and a preprocessed
   look-up table can be used to determine in constant time the structure of
   the piece of the straight skeleton within each subvoxel.
   The complexity of the skeletal piece (within a voxel) is also constant.
   Figures~\ref{F-voxel-skeletons}(a--h)
   \begin{figure}[t]
      \centering
      \begin{tabular}{cccccc}
         ~\scalebox{0.3}{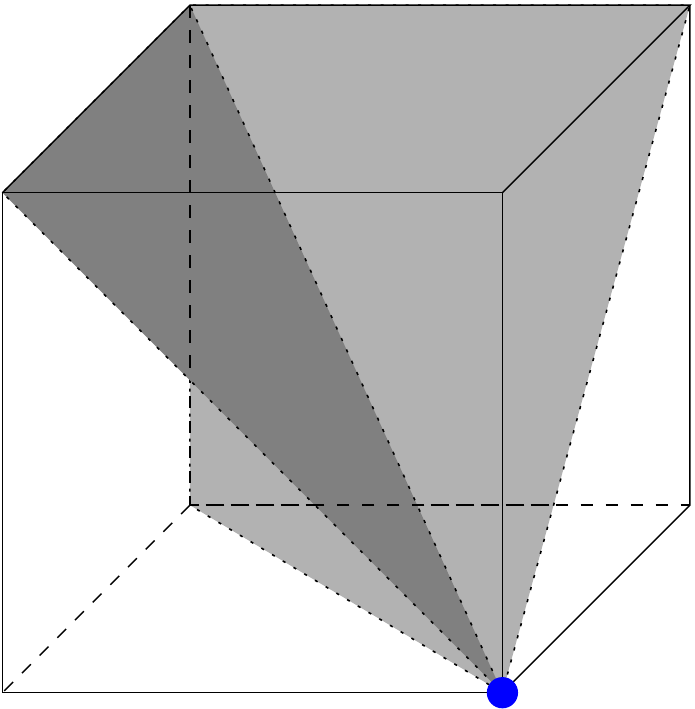}~ &
            ~\scalebox{0.3}{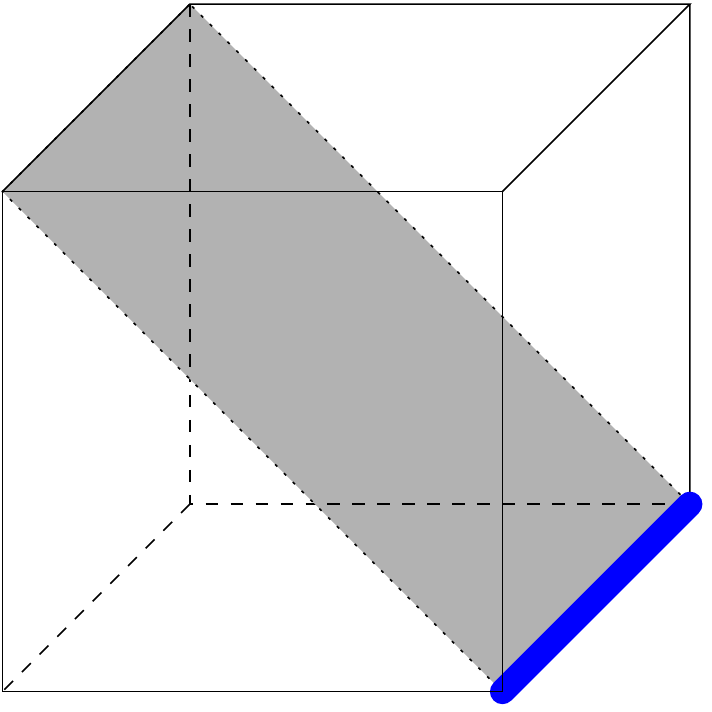}~ &
            ~\scalebox{0.3}{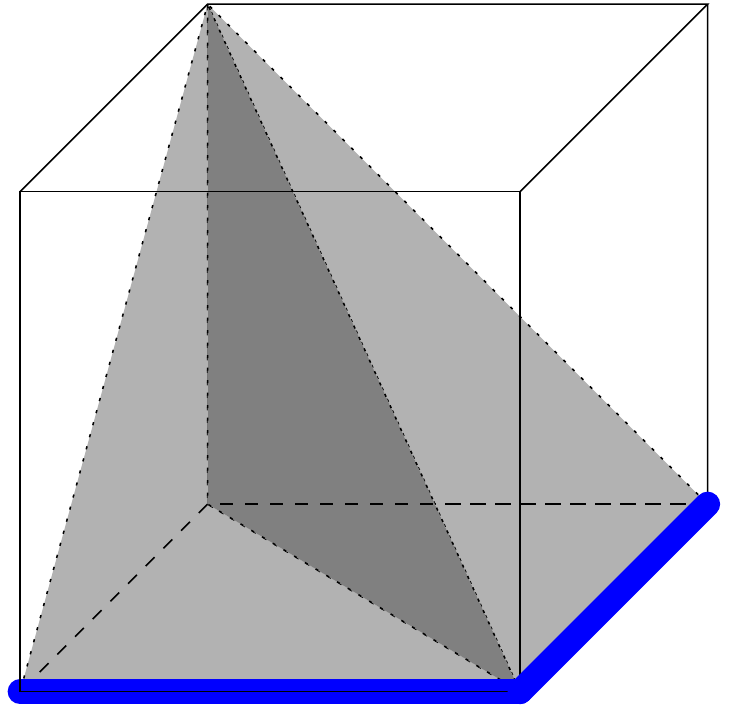}~ &
            ~\scalebox{0.3}{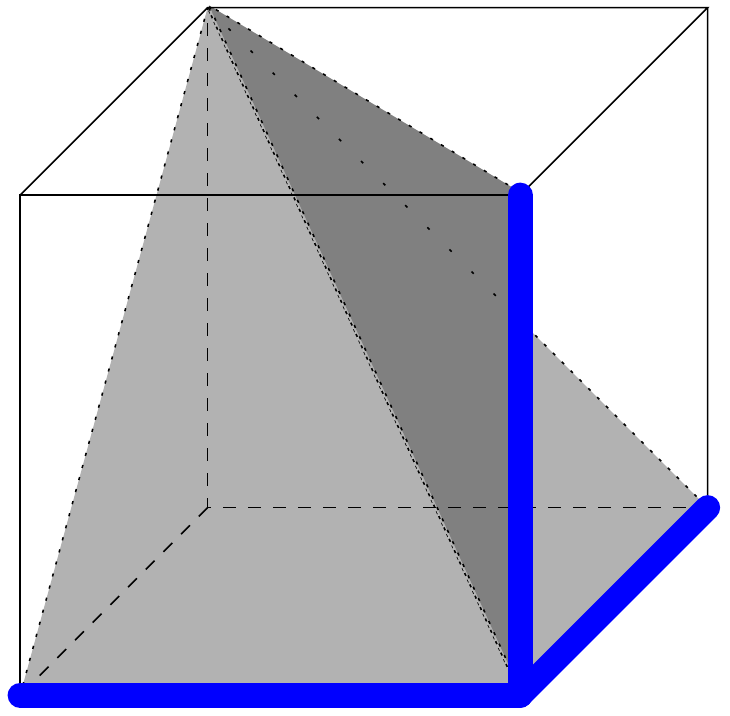}~ &
            ~\scalebox{0.3}{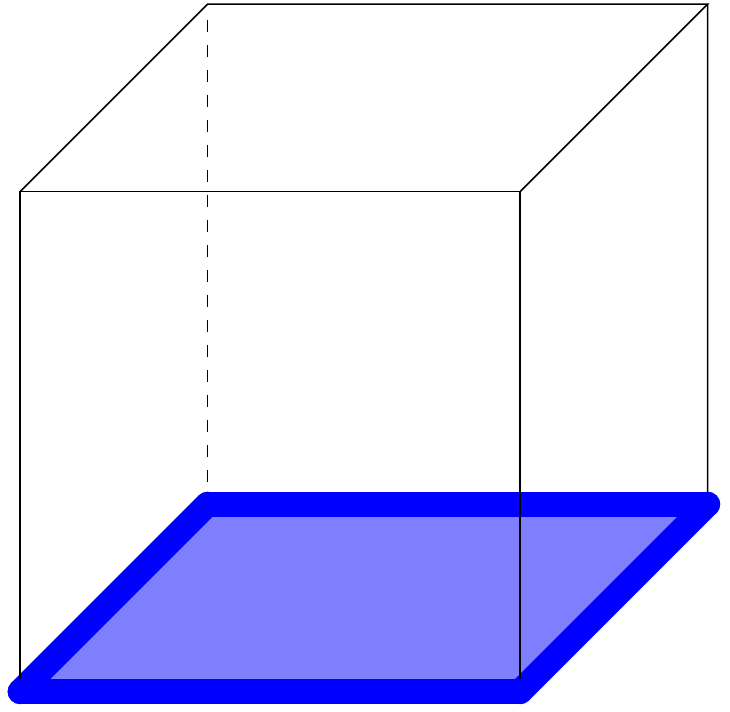}~ &
            ~\scalebox{0.3}{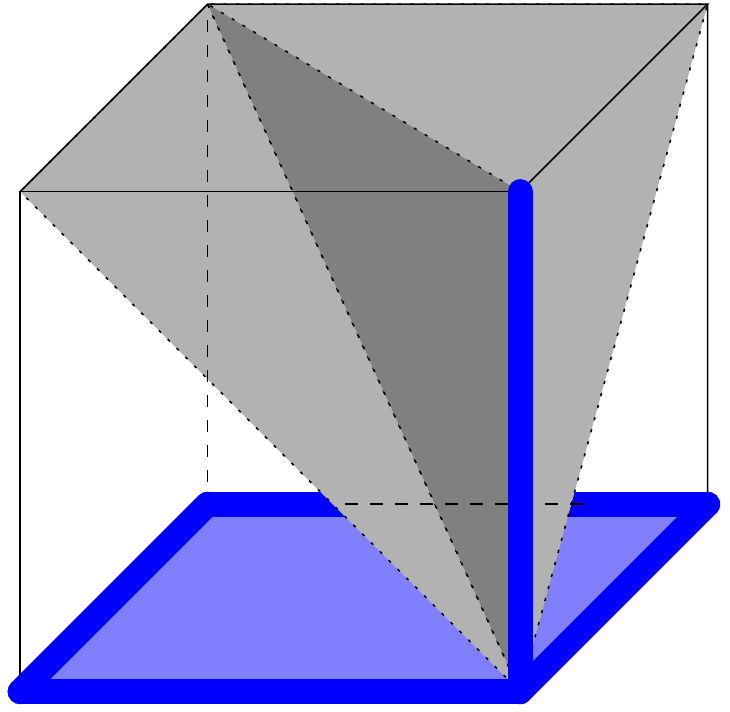}~ \\
         (a) Vertex & (b) Edge & (c) Two edges &
         (d) Three edges & (e) Face & (f) Face and edge \\
      \end{tabular}
      \begin{tabular}{ccccc}
            ~\scalebox{0.3}{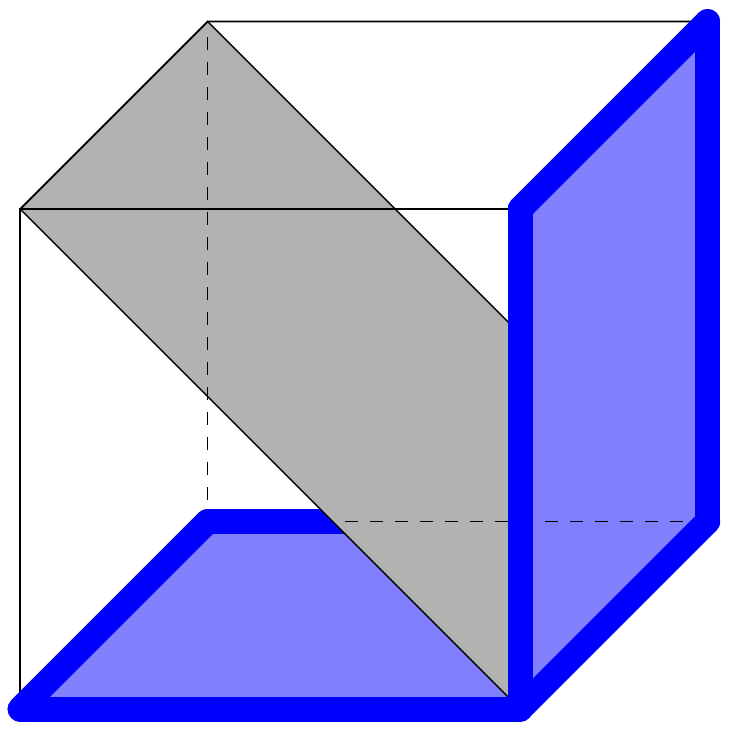}~ &
            ~\scalebox{0.3}{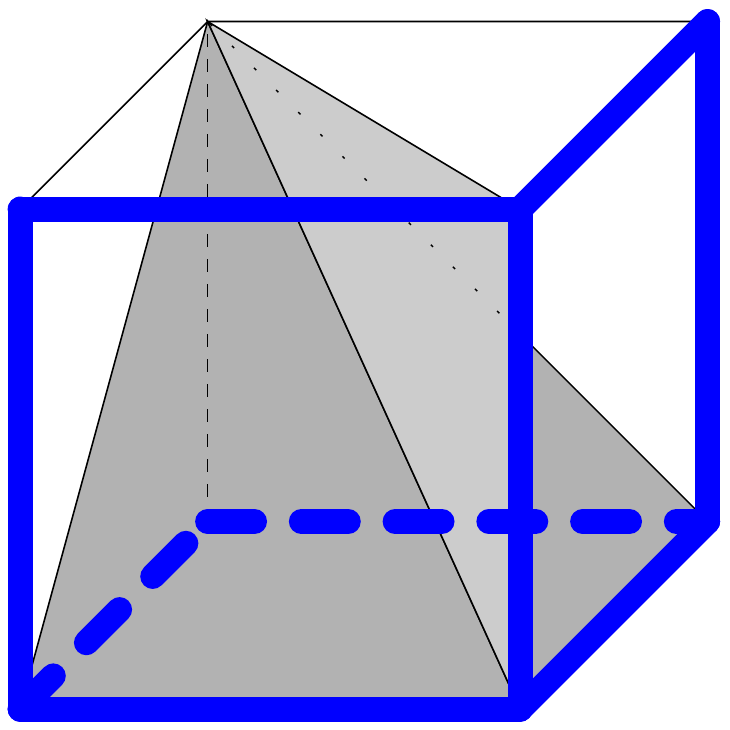}~ &
         \scalebox{0.3}{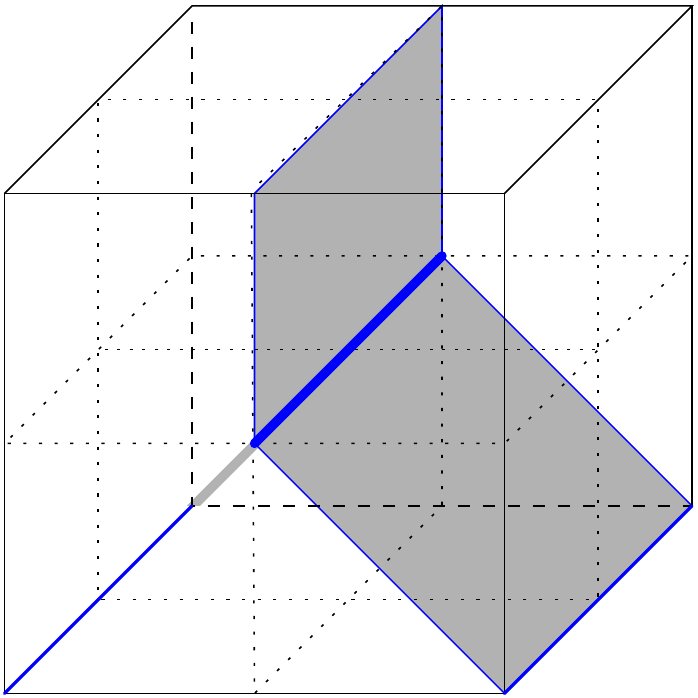} &
            \scalebox{0.3}{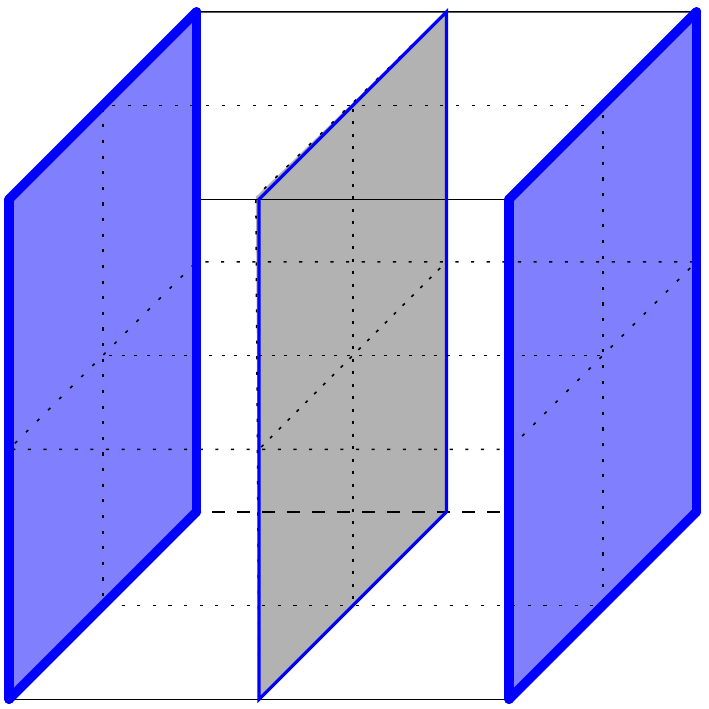} &
            \scalebox{0.3}{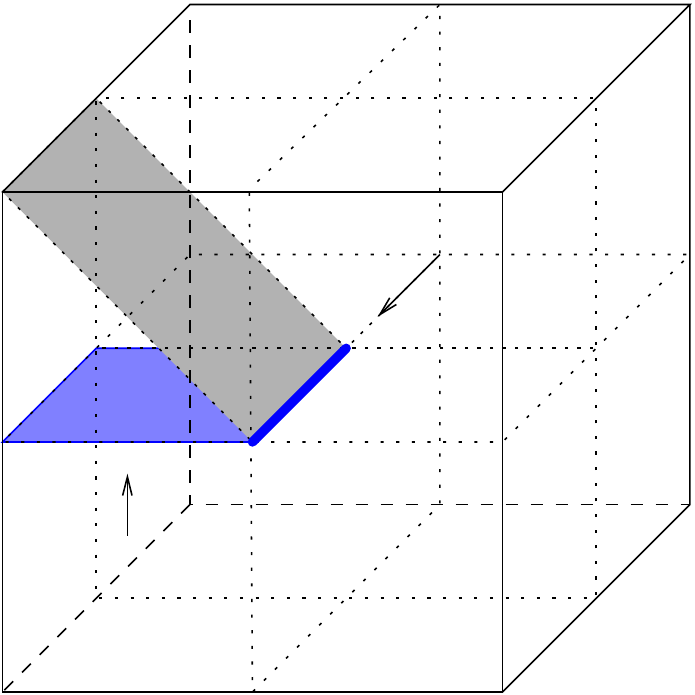} \\
          (g) Two faces &
            (h) Three faces &
         ~(i) Overlapping edges~ &
            ~(j) Overlapping faces~ &
            ~(k) Overlapping edge and face~ \\
      \end{tabular}
      \caption{Cases of straight skeleton within a subvoxel (a-h) or voxel (i-k).}
      \label{F-voxel-skeletons}
   \end{figure}
   show the creation of a skeletal piece in the interior of a subvoxel in
   simple cases.  The features through which the moving fronts enter the
   subvoxel are shown enlarged.
   Figures~\ref{F-voxel-skeletons}(i--k) show (in full voxels) a
   few cases of overlapping entry features.
\ifFullversion%
   In (i), two skeletal pieces ($a$ and $b$) emanate diagonally upward,
   meeting in one edge; the continuation of the skeleton is the piece $c$.
   This figure also models a different case, in which the skeletal pieces
   $a$ and $c$ meet at the thick edge.  In this case the continuation of the
   skeleton is the piece $b$.  In (j), two fronts move horizontally toward
   each other, and meet in one face which become a skeletal piece.
   In (k), a face moves upward vertically, meeting a concave edge which moves
   down diagonally; the continuation of the skeleton is as shown.
\fi

   To recap, the algorithm processes all voxels in layers, in a total
   of $\Theta(V)$ voxel operations, each of which takes constant time and
   contributes a constant amount of skeletal features.
   The algorithm terminates when there are no more voxels to process and the entire straight skeleton of the polycube has been computed.
\ifFullversion%

   One way to see why the skeletal pieces constructed within
   neighboring eighths of a voxel (belonging to the same original voxel) are
   always ``glued'' together consistently without leaving any
   improperly connected dangling skeletal pieces is by imagining that we handle whole voxels at a time,
   processing all eight subvoxels simultaneously by using a much larger
   look-up table.  An alternative argument is
   that different subvoxels of the same voxel are not independent---they are
   hit by the same moving front, either at the same time or with a delay of
   one half of a time unit.
\fi%
\end{proof}

\begin{figure}[t]
\begin{minipage}[b]{.35\textwidth}
   \centering\includegraphics[width=0.6\textwidth]{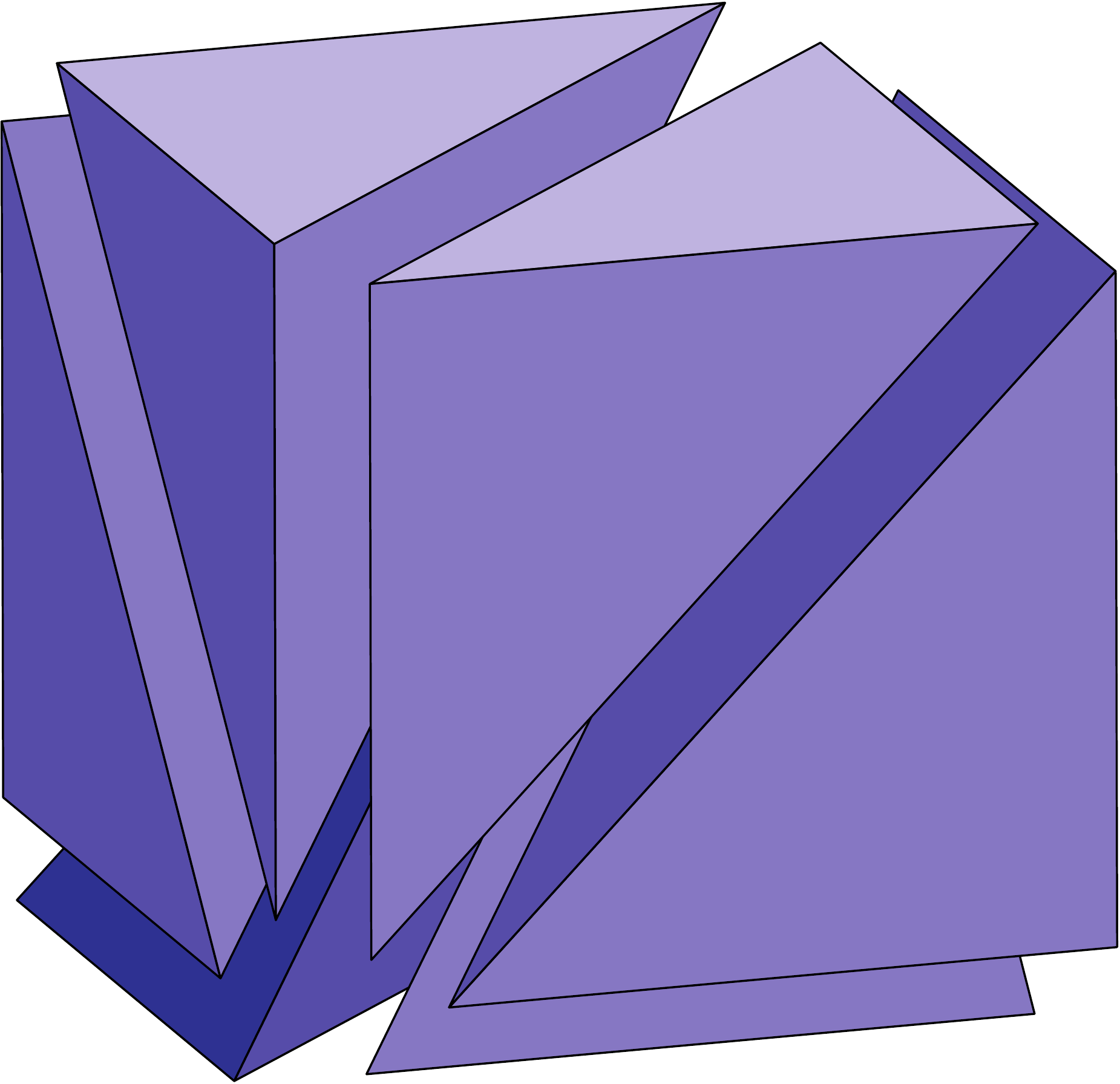}
   \caption{Partitioning a subvoxel into tetrahedra.}
   \label{F-sixcube}
\end{minipage}
\hfill
\begin{minipage}[b]{.645\textwidth}
   \centering
   \begin{tabular}{cc}
      \scalebox{0.55}{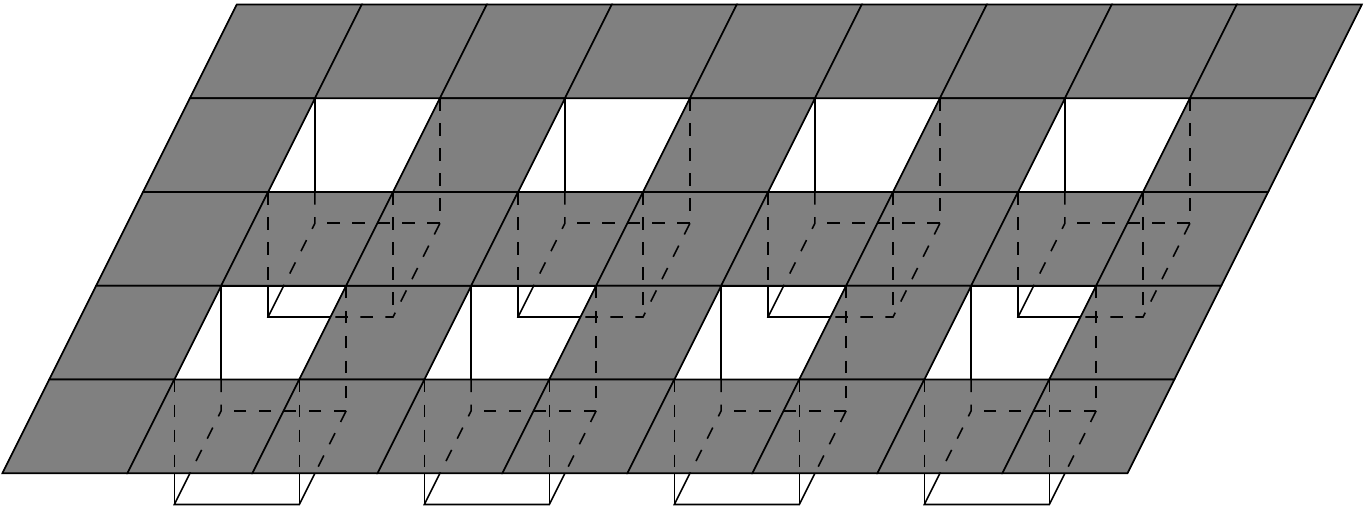} &
         \scalebox{0.3}{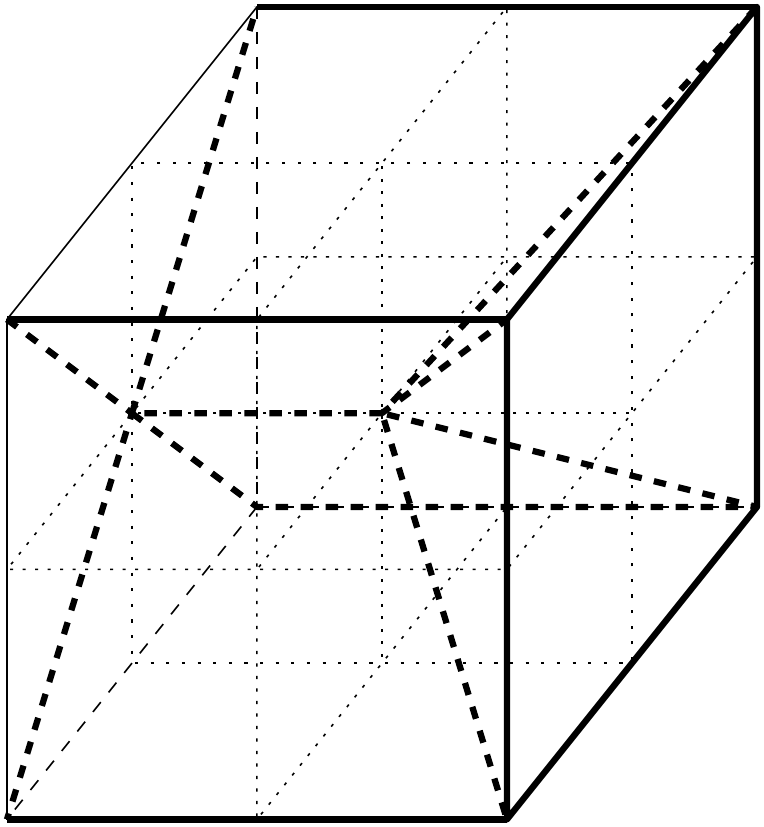} \\
      (a) & (b)
   \end{tabular}
   \caption{A polycube of volume $V$ whose skeleton has complexity $\Theta(V)$.}
   \label{F-theta-v-skeleton}
\end{minipage}
\end{figure}

A unified way to look at all cases above is to partition a voxel, as above, to
eight subvoxels, and then partition each subvoxel into six tetrahedra each of which
is the convex hull of one of the six three-edge paths connecting the
subvoxel's integer vertex with its half-integer vertex (Figure~\ref{F-sixcube}).
Thus, every voxel is partitioned into~48 tetrahedra.
All skeletal cells are unions of these tetrahedra, and the surface
of the skeleton is composed of their boundary triangles.
By maintaining ``visited'' marks on the tetrahedra and on the integer and
half-integer vertices, one can sweep the wavefront and compute the
revealed pieces of the skeleton.

Many simple examples show that the sweeping algorithm is worst-case
optimal up to constant factors, since in the worst case the complexity of a polycube made of $V$
voxels is $\Theta(V)$.  One such example, shown in
Figure~\ref{F-theta-v-skeleton}(a), is made of a flat layer of cubes (not shown),
with a grid of supporting ``legs,'' each a single cube.  Thus, the
number of legs is about one fifth of the total number of voxels.
The skeleton of this object has features within every leg, as shown
in Figure~\ref{F-theta-v-skeleton}(b) (the bottom of a leg corresponds to
the right side of the figure).

\subsection{Output-Sensitive Voxel Sweep}

\label{S-o-s-voxel}

The straight skeleton of a polycube, as constructed by the previous algorithm, contains features within some voxels, but other voxels may not participate in the skeleton; nevertheless, the algorithm must consider all voxels and pay in its running time for them. In this section we outline a more efficient algorithm that computes the straight skeleton in time proportional only to the number of voxels containing skeleton features, or equivalently, in time proportional to the \emph{surface area} of the straight skeleton rather than its
volume\ifFullversion
(since the surface area of the portion of the straight skeleton of an orthogonal
polygon contained in any voxel is $O(1)$)
\fi. 
Necessarily, we assume that the input polycube is provided as a space-efficient boundary representation rather than as a set of voxels, for otherwise simply scanning the input would take more time than we wish to spend.

Our algorithm consists of an outer loop, in which we advance the moving front
of the polycube boundary one time step at a time\ifFullversion (as in the
previous algorithm)\fi, and an inner loop, in which we capture all features of the straight skeleton formed in that time step. During the algorithm, we maintain at each step a representation of the moving front, as a collection of polygons having orthogonal and diagonal edges. As long as each operation performed in the inner and outer loops of the algorithm can be charged against straight skeleton output features, the total time will be proportional to the output size.

In order to avoid the randomization needed for hashing, several steps of our algorithm will use as a data structure a direct-addressed lookup table, which we summarize in the following lemma:

\begin{lemma}
\label{lem:direct-address}
In time proportional to the boundary of an input polycube, we may initialize a data structure that can repeatedly take as input a collection of objects, indexed by integers within the range of coordinate values of the polycube vertices, and produce as output a graph, the vertices of which are sets of objects that have equal indices and the edges of which are pairs of sets with index values that differ by one. The time per operation is proportional to the number of objects given as input.
\end{lemma}

\begin{proof}
We use an array, indexed by the given integer values, containing a list of objects in each array cell. Initially, we set all lists to empty. To handle a given collection of objects, we place each object in the list given by the object's index, and create a list $L$ of nonempty index values as we do so; each time we add an object to an empty list, we add that list's index to $L$. We then create a graph having as its vertices the lists indexed by $L$; for each vertex we search the array for the two adjacent indices and create the appropriate graph edges. Finally, we use $L$ to replace each nonempty list of the array with a new empty list.
\end{proof}

In more detail, in each step of the outer loop of the algorithm, we perform the following steps:
\begin{enumerate}
\item Advance each face of the wavefront one unit inward. In this advancement process, we may detect events in which a wavefront edge shrinks to a point, forming a straight skeleton vertex. However, events involving pairs of features that are near in space but far apart on the wavefront may remain undetected. Thus, after this step, the wavefront may include overlapping pairs of coplanar oppositely-moving faces.
\item For each plane containing faces of the new wavefront boundary, detect pairs of faces that overlap within that plane, and find the features in which two overlapping face edges intersect or in which a vertex of one face lies in the interior of another face. This step can be performed as a sequence of smaller steps:
\begin{itemize}
\item Group coplanar faces of the wavefront using the data structure of Lemma~\ref{lem:direct-address}.
\item Within each plane $P$, form a set $S_P$ of the wavefront edges intersected with each voxel. We assume that the plane is parallel to the $xy$ plane; the $xz$ and $yz$ cases are handled symmetrically.
\item For each plane $P$, use Lemma~\ref{lem:direct-address} to form a graph $G_P$; vertices in $G_P$ represent sets of edges in $S_P$ with the same left $x$-coordinate, and edges in $G_P$ connect sets with consecutive $x$-coordinates. The connected components of this graph are paths representing subsets of wavefront features that might possibly interact with each other, sorted by their $x$-coordinates.
\item Within each connected component of each graph $G_P$, use the sorted order to perform a plane sweep algorithm that finds segment intersections and locates the face containing each vertex. Report as straight skeleton events each intersection between edges of different boundary faces and each vertex that belongs to a boundary face other than the one on which it is a boundary vertex.
\end{itemize}
\item In the inner loop of the algorithm, propagate straight skeleton features within each face of the wavefront from the points detected in the previous step to the rest of the face. If two faces overlap in a single plane, the previous step will have found some of the points at which they form straight skeleton vertices, but the entire overlap region will form a face of the straight skeleton. We propagate outward from the detected intersection points using depth-first-search, voxel by voxel, to determine the straight skeleton features contained within the overlap region.
\end{enumerate}

\noindent In summary, we have:

\begin{theorem}
One can compute the straight skeleton of a polycube in time
proportional to its surface area.
\end{theorem}

\section{Orthogonal Polyhedra}

We consider here a more general class of inputs than voxels: \emph{orthogonal polyhedra} in which all faces are parallel to two of the coordinate axes.

\subsection{Definition}

As in the two-dimensional case, we define the straight skeleton of an orthogonal polyhedron $P$ by a continuous shrinking process in which a sequence of nested ``offset surfaces'' are formed, starting from the boundary of the given polyhedron, with each face moving inward at a constant speed.
At time $t$ in this shrinking process, the offset surface $P_t$ for $P$ consists of the set of points at $L_\infty$ distance exactly $t$ from the boundary of $P$. For almost all values of $t$, $P_t$ will itself be a polyhedron, but at some time steps $P_t$ may have a non-manifold topology, possibly including flat sheets of surface that do not bound any interior region. When this happens, the evolution of the surface undergoes sudden discontinuous changes, as these surfaces vanish at time steps after $t$ in a discontinuous way. To make this notion of discontinuity more precise, we define a \emph{degenerate point} of $P_t$ to be a point $p$ that is on the boundary of $P_t$, such that, for some $\delta$, and all $\epsilon>0$, $P_{t+\epsilon}$ does not contain any point within distance $\delta$ of $p$. Equivalently, a degenerate point is a point of $P_t$ that does not belong to the closure of the interior of $P_t$.

At each step in the shrinking process, we imagine the surface of $P_t$ as \emph{decorated} with \emph{seams} left over when sheets of degenerate points occur. To be more specific, suppose that $P$ contains two disjoint faces, both parallel to the $xy$ plane at the same $z$-height; then, as we shrink $P$, the corresponding faces of $P_t$ may grow toward each other, eventually meeting. When they do meet, they leave a seam between them. Seams can also occur when two parts of the same nonconvex face grow toward and meet each other. After a seam forms, it remains on the face of $P_t$ on which it formed, orthogonal to the position at which it originally formed.

We may also describe these seams in a more intrinsic, static way. Let $\Pi$ be any axis-aligned plane containing a face or faces of $P$, and let $S_\Pi$ be the two-dimensional straight skeleton in $\Pi$ of the exterior of these faces, not including the straight skeleton edges that touch the vertices of $P$. Then the decoration on any face $f$ of $P_t$, corresponding to a face of $P$ belonging to plane $\Pi$, is formed by translating $S_\Pi$ orthogonally into the plane of $f$ and intersecting it with $f$.

We define the \emph{straight skeleton} of $P$ to be the union of three sets:
\begin{enumerate}
\item The points that, for some time step $t$, belong to an edge or vertex of $P_t$.
\item The degenerate points for $P_t$ for some time step $t$.
\item The points that, for some time step $t$, belong to a seam of $P_t$.
\end{enumerate}

The straight skeleton may be viewed as a cell complex in $\R^3$, consisting
of \emph{faces} (maximal subsets of points that have a 2D neighborhood in the straight skeleton), \emph{edges} (maximal line segments of points that either do not lie in a face, lie on the boundary of a face, or lie in the intersection of two or more faces), and \emph{vertices} (endpoints of edges).

\subsection{Complexity Bounds}

As each face has at least one boundary edge, and each edge has at least one vertex, we may bound the complexity of the straight skeleton by bounding the number of its vertices. Each vertex corresponds to an \emph{event}, that is, a point $p$ in space (the location of the vertex), the time $t$ for which $p$ belongs to the boundary of $P_t$, and the set of features of $P_{t-\epsilon}$ near $p$ for small values of $\epsilon$ that contribute to the event.
We may classify events into six types.

\begin{description}
\item[Concave-vertex events] describe the situation in which one of the features of $P_{t-\epsilon}$ involved in the event is a \emph{concave vertex}: that is, a vertex of $P_{t-\epsilon}$ such that seven of the eight quadrants surrounding that vertex lie within $P_{t-\epsilon}$. In such an event, this vertex must collide against some oppositely-moving feature of $P_t$.
\item[Reflex-reflex events] describe events that are not concave-vertex events, but in which the event involves the collision between two components of boundary of $P_{t-\epsilon}$ that prior to the event are far from each other as measured in geodesic distance around the boundary, both of which include a reflex edge. These components may either be themselves a reflex edge, or a vertex that has a reflex edge within its neighborhood.
\item[Reflex-seam events] describe events that are not either of the above two types, but in which the event involves the collision between two different components of boundary of $P_{t-\epsilon}$, one of which includes a reflex edge. The other boundary component must necessarily be a seam edge or vertex, because it is not possible for a reflex edge to collide with a convex edge of $P_{t-\epsilon}$ unless both edges are part of a single boundary component.
\item[Seam-seam events] in which vertices or edges on two seams, on oppositely oriented parallel faces of $P_{t-\epsilon}$, collide with each other.
\item[Seam-face events] in which a seam vertex on one face of $P_{t-\epsilon}$ collides with a point on an oppositely oriented face that does not belong to a seam.
\item[Single-component events] in which the boundary points near $p$ in $P_{t-\epsilon}$ form a single connected subset.
\end{description}

\begin{theorem}
The straight skeleton of an $n$-vertex orthogonal polyhedron has complexity $O(n^2)$.
\end{theorem}

\begin{proof}
We count the events of each different type.
Each concave-vertex event is the final event involving its concave vertex,
and no event creates any new concave vertex; therefore, there are $O(n)$ such
events. Each reflex edge of $P_t$ corresponds to a reflex edge of $P$, so
each reflex-reflex event of $P_t$ can be charged against a pair of reflex
edges of $P$; each such pair yields at most one reflex-reflex event, so the
total number of such events of this type is $O(n^2)$. Similarly, we may charge reflex-seam events to a pair of a reflex edge of $P$ and an edge of some $S_\Pi$, and each seam-seam event to a pair of two edges of $S_\Pi$ and $S_{\Pi'}$; there are $O(n^2)$ such pairs. Each seam-face event is the final event involving a vertex of $S_\Pi$, so there are $O(n)$ such events. Finally, each single-component event involves at least one edge of $P_{t-\epsilon}$ that is bounded by two oppositely-oriented face planes, and shrinks to nothing in $P_t$; these events reduce the total number of edges in $P_t$, and can be charged against the events of other types that created those edges.
\end{proof}

\subsection{Algorithms}

The view of straight skeletons as generated by a moving surface that changes combinatorially at a sequence of discrete events may also be used as the basis of an algorithm for constructing the skeleton of a given orthogonal polyhedron. It is straightforward to determine in constant time the changes to $P_t$ resulting from an event at time $t$, and to construct the corresponding straight skeleton features, so the problem reduces to determining efficiently the sequence of events that happen at different times in the evolution of this moving surface, and distinguishing actual events from combinations of features that could generate events but don't.

We provide two algorithms for solving this event generation problem, and,
therefore, for constructing straight skeletons, of incomparable complexities.

\begin{theorem}
There is a constant $c$ such that the straight skeleton of an orthogonal polyhedron with $n$ vertices and $k$ straight skeleton features may be constructed in time $O(k\log^c n)$.
\end{theorem}

\begin{proof}
Each event in our classification (except for the single-component events, which may be handled by a simple event queue) is generated by the interaction of two features of the moving surface $P_t$: pairs of edges or seams in most of the events, pairs of a vertex and a face in some of them.
To generate these events, ordered by the time at which they occur, we use a data structure of Eppstein~\cite{e-demst-95,Epp-JEA-00} for maintaining a set of items and finding the pair of items minimizing some binary function $f(x,y)$; in our application, $f(x,y)$ is the time at which an event is generated by the interaction of items $x$ and $y$, or $+\infty$ for two items that do not interact. The data structure reduces this problem (with polylogarithmic overhead) to a simpler problem: maintain a dynamic set $X$ of items, and answer queries asking for the first interaction between an item $x$ in $X$ and a query item $y$. We need separate first-interaction data structures of this type for edge-edge, vertex-face, and face-vertex interactions. \ifFullversion(Note that, although vertex-face and face-vertex problems are equivalent in terms of the problem of finding the first interacting pair, they are different in terms of the problem of finding the first interaction for a query item $y$, and the reduction involves both versions of the problem.)\fi

To handle the edge-edge interactions, we first partition the edges into
finitely-many equivalence classes by their orientations and by the velocities
at which their endpoints move as $P_t$ evolves, and treat each equivalence
class separately. Within an equivalence class, each edge can be described by
four coordinates\ifFullversion (the position of one of its endpoints and its
length)\fi, so the first-interaction problem can be handled as an appropriate four-dimensional orthogonal range searching problem.

To handle the vertex-face and face-vertex interactions, we need to reduce the faces (which may be complicated planar objects with holes) to regions with bounded description complexity, so that we may again employ orthogonal range searching techniques. To do so, we first partition each planar straight skeleton $S_\Pi$ into regions, where each region is either a face of the input that lies within plane $\Pi$ or the straight skeleton region belonging to one of the input edges. Next, we further partition $S_\Pi$ into trapezoids using a vertical visibility decomposition. As $t$ changes and $P_t$ evolves, each of these trapezoids will move perpendicular to plane $\Pi$, and in addition, its edges may move linearly outward or inward depending on the face structure of $P$ near that face. Additionally, some of these trapezoids may become partially or completely blocked from participating in the boundary of $P_t$, due to other faces that interact with them; however, in our vertex-face and face-vertex interaction data structures, we ignore this blocking effect, as whenever some trapezoid is blocked it is due to some other boundary feature being closer to any objects that might interact with the trapezoid. With this decomposition, and a partition of the input objects into finitely many subclasses according to their shape and velocity, we have a set of objects that can be specified with finitely many dimensions (three for each vertex, five for each trapezoid) to which we may apply an appropriate orthogonal range searching data structure.
\end{proof}

Although within a polylogarithmic factor of optimal, this algorithm may be complex and difficult to implement. If we wish to achieve worst-case optimality rather than output-sensitive optimality, a much simpler algorithm is possible.

\begin{theorem}
The straight skeleton of an orthogonal polyhedron with $n$ vertices and $k$ straight skeleton features may be constructed in time $O(n^2\log n)$.
\end{theorem}

\begin{proof}
For each pair of objects that may interact (features of the input polyhedron $P$ or of the two-dimensional straight skeletons $S_\Pi$ in each face plane $\Pi$), we compute the time at which that interaction would happen. We sort the set of pairs of objects by this time parameter, and process pairs in order; whenever we process a pair $(x,y)$, we consult an additional data structure to determine whether the pair causes an event or whether the event that they might have caused has been blocked by some other features of the straight skeleton.

To test whether an edge-edge pair causes an event, we maintain a binary search tree for each edge, representing the family of
\ifFullversion line \fi
segments into which the line containing that edge (translated according to the motion of the surface $P_t$) has been subdivided in the current state of the surface $P_t$. An edge-edge pair causes an event if the point at which the event would occur currently belongs to line segments from the lines of both edges, which may be tested in logarithmic time.

To test whether a vertex-face pair causes an event, we first check whether the vertex still exists at the time of the event, and then perform a point location query to locate the point in $S_\Pi$ at which it would collide with a face of the input belonging to plane $\Pi$. The collision occurs if the orthogonal distance within plane $\Pi$ from this point to the nearest input face is smaller than the time parameter at which the collision would occur. We do not need to check whether some other features of the straight skeleton might have blocked features of $S_\Pi$ from belonging to the boundary of $P_t$, for if they did they would also have led to some earlier vertex-face event causing the vertex to be removed from $P_t$.

Thus, each object pair may be tested using either a dynamic binary search tree or a static point location data structure, in logarithmic time per pair.
\end{proof}

\section{General Polyhedra}

\label{Sec:general}

\ifFullversion
In this section we consider the straight skeleton of a general polyhedron.
\fi

\subsection{Ambiguity}

Defining the straight skeleton of a general 3D polyhedron is inherently ambiguous, unlike the cases for convex and orthogonal polyhedra.
The ambiguity stems from the fact that, whereas convex polyheda are
defined uniquely by the planes supporting their faces, nonconvex polyhedra
are defined by \emph{both} the supporting planes and a given topology, which
is not necessarily unique.  Thus,
\ifFullversion during the offsetting process, \else while being offset, \fi
a polyhedron can propagate from a given
\ifFullversion intermediate (or initial) \fi
state into multiple equally valid topological configurations.
This issue was alluded to
\ifFullversion in a paper \fi
by Demaine {\it et al.}~\cite{HingedPolyforms3D_WADS2005}\ifFullversion,
in fact, in a reference to a private communication by Jeff Erickson\fi.

We make the nature of this ambiguity more precise in
Figure~\ref{F-ambiguity}(a).
\begin{figure}
   \centering
   \begin{tabular}{cccccc}
      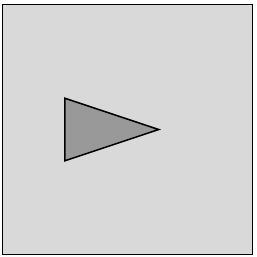 &
         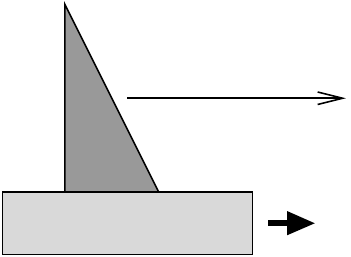 & ~ &
         \includegraphics[width=1.2In]{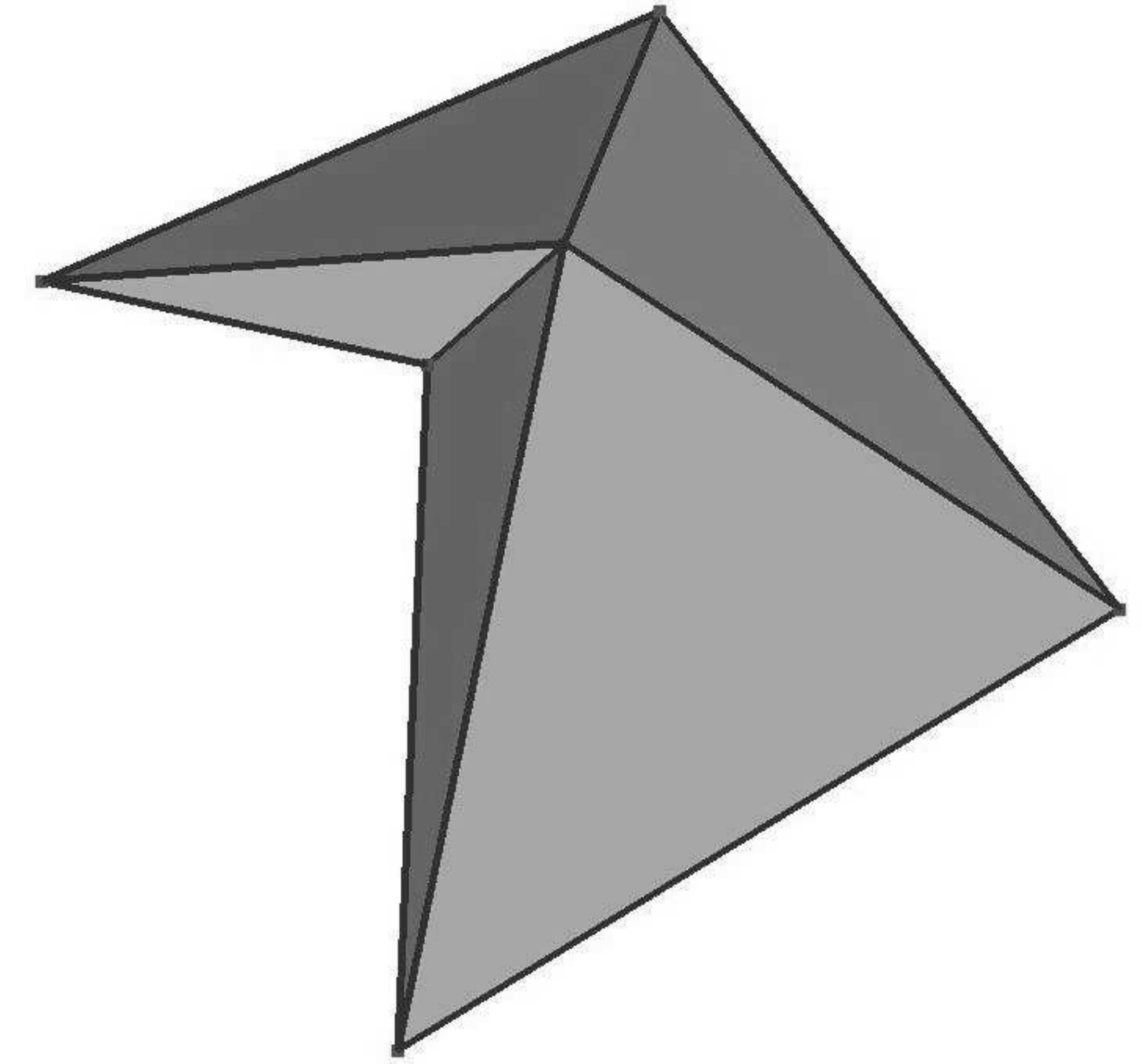} &
         \includegraphics[width=1.2In]{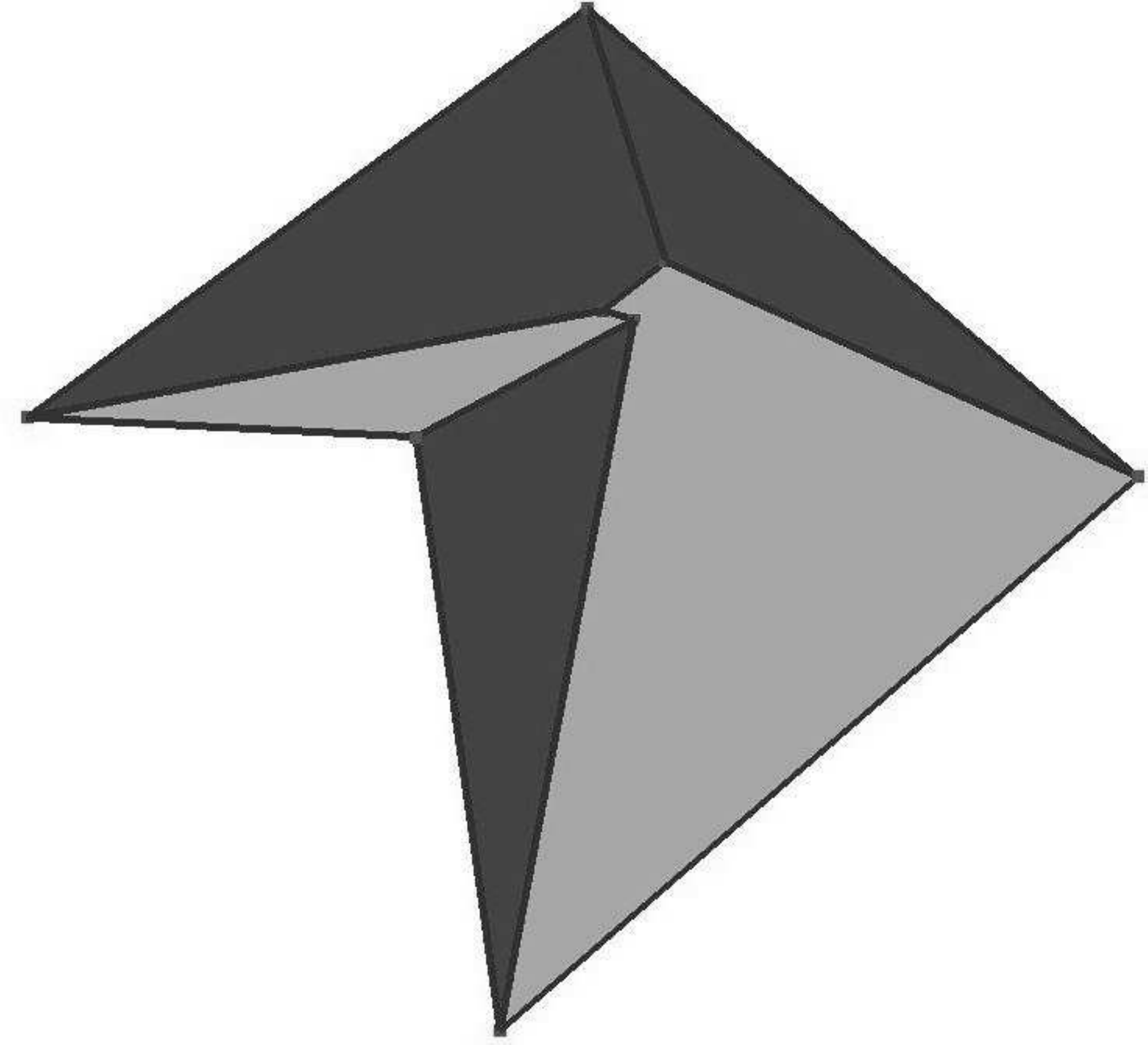} &
         \includegraphics[width=1.2In]{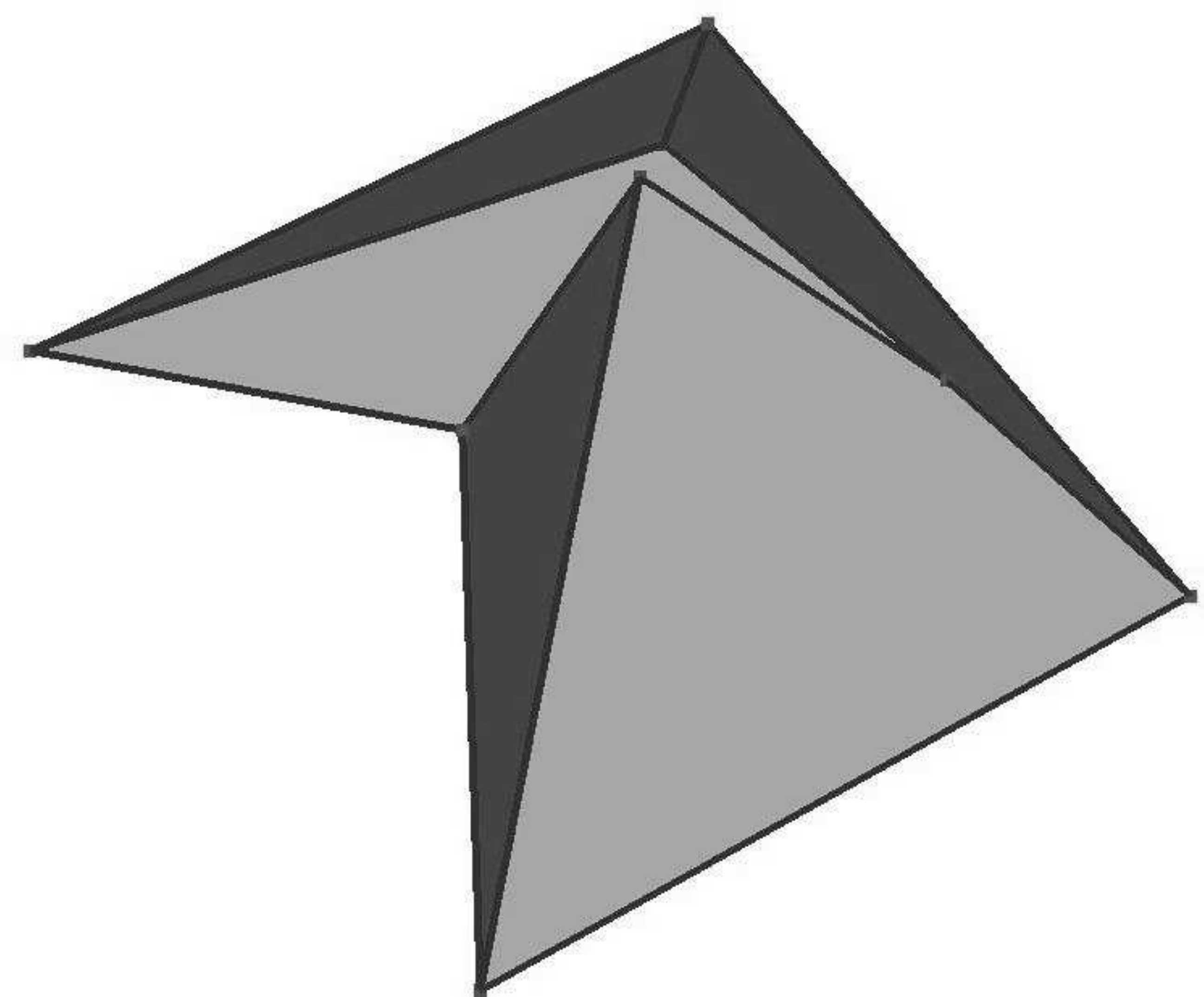} \\
      Top view & Side view ~~~~~ & &
         Initial topology & Our method & Another solution \\
      \multicolumn{2}{c}{(a) A Simple example} & &
         \multicolumn{3}{c}{(b) A more complex example}
   \end{tabular}
   \caption{3D skeleton ambiguity.}
   \label{F-ambiguity}
\end{figure}
The problem is illustrated
with respect to two pieces of skeleton---a wedge, $A$, and a tabletop,
$B$---that are growing relative to each other.
Because of the angle of the two front planes of the wedge, the growing wedge
\ifFullversion
is on a trajectory to eventually grow
\else
eventually grows
\fi
past the tabletop.
The issue that arises at this point is 
to determine how the wavefronts should continue growing.
\ifFullversion
There are several choices (in fact, an infinite number of choices).
For example, the first three coauthors on this paper (in no
particular order) respectively advocated the three following
resolutions (with the fourth coauthor being too wise to agree with any
of them, suggesting a fourth solution described below): 
\else
There are several choices, for example:
\begin{itemize}
\item The front end of the wedge $A$ is blunted by clipping it with
      the plane defined by the side of the tabletop.
\item The wedge continues growing forward, but is blocked from growing
      downward by clipping it with the plane defined by the top of the
      tabletop.
\item The wedge suddenly projects into the empty space in front of the
      table and continues growing out from there.
\end{itemize}

There are other possibilities, as well.
In fact, all three suggestions listed above cause a contradiction or
a noncontinuous propagation of the wavefront in certain cases.
One poor choice, however, not listed above, is to allow the wedge
$A$ to grow through to the other side of $B$ in the case that $A$ reaches
the edge of $B$ and moves past the edge.
With this 3D skeleton definition it is possible to construct
self-contradictory examples
with three wedges, $A_1$, $A_2$, and $A_3$, such that $A_1$ and $A_3$
are on opposite sides of the tabletop and oriented in a way that
if $A_1$ breaks through $B$, then it blocks $A_2$, which in turn does
not block $A_3$, which in turn breaks through $B$ and prevents $A_1$
from breaking through $B$.
Likewise, if $A_1$ doesn't break through $B$, then it doesn't block
$A_2$, which blocks $A_3$, which, in turn doesn't block $A_1$, which
breaks through $B$.
Thus, we can at least conclude that this rule is an
inappropriate choice\ifFullversion for resolving ambiguities in
the definition of general 3D straight skeletons\fi.

A more general example of the inherent ambiguity of the propagation of the
straight skeleton is shown in Figure~\ref{F-ambiguity}(b).  The figure
shows a vertex of degree~5, and two possible topologies during the
propagation.  This is the so-called weighted-rooftop problem:
Given a base polygon and slopes of walls, all sharing one vertex, determine
the topology of the rooftop of the polygon, which does not always have a
unique solution. In our definition of the skeleton, we define a consistent
method for the initial topology and for establishing topological changes
while processing the algorithm's events, based on the two-dimensional
weighted straight skeleton.
This method is described in Section~\ref{SS-gen-algorithm}.

\subsection{A Combinatorial Lower Bound}

\ifFullversion
We first address the convex case.  As mentioned in the introduction,
Held~\cite{Held94} showed that the complexity of the straight skeleton
of an $n$-vertex convex polyhedron is
$\Omega(n^2)$ (see Figure~\ref{fig:QuadraticConvex}).
\begin{figure}
   \subfigure[A convex double ``shell'']{
   \label{fig:ConvexNoSkeleton}
   \includegraphics[width=2In]{convex_no_skeleton}}
   \subfigure[The quadratic-size straight skeleton in black lines]{
   \label{fig:QuadraticConvex1}
   \includegraphics[width=2In]{quadratic_convex1}}
   \subfigure[An upper view of the skeleton]{
   \label{fig:QuadraticConvex2}
   \includegraphics[width=2In]{quadratic_convex2}}
   \caption{A worst-case-complexity skeleton in the convex case.}
   \label{fig:QuadraticConvex}
\end{figure}
The following theorem shows a matching upper bound.  This, we establish a
tight bound of $\Theta(n^2)$ on the complexity of the straight skeleton in
the worst case.

\begin{theorem}
   Let $n$ be the number of faces of a polyhedron. Then, the straight
   skeleton of the polyhedron contains exactly $n$ 3-cells and $O(n^2)$
   faces, edges, and vertices.
\end{theorem}

\begin{proof}
   The straight skeleton of the polyhedron contains $n$ 3-cells because
   each face $f$ of the polyhedron sweeps exactly one 3-cell of the skeleton,
   which is monotone with respect to $f$.  Since the polyhedron is convex,
   no face of it is split during the propagation of the polyhedron's boundary.
   In addition, every 3-cell is convex.
   Therefore, every 3-cell of the polyhedron, during the entire offsetting
   process, shares at most one skeletal face with any other 3-cell (otherwise
   faces of the polyhedron would be split, or 3-cells would not be convex).
   Thus, there are $O(n^2)$ skeletal faces, where each 3-cell is bounded by
   $O(n)$ of them.  Hence, every 3-cell is bounded by $O(n)$ skeletal edges
   and vertices, for a total of $O(n^2)$.
\end{proof}
\fi

We now show that the 3D straight skeleton of an $n$-vertex
general simple polyhedron can have
asymptotic combinatorial complexity strictly greater than the
complexity of the 3D straight skeleton of an orthogonal polyhedron.

\begin{theorem}
   The combinatorial complexity of a 3D skeleton for a simple
   polyhedron is $\Omega(n^2\alpha^2(n))$ in the worst case,
   where $\alpha(n)$ is the inverse of the Ackermann function.
\end{theorem}

\begin{proof} (Sketch)
   We begin by showing that the cross-section
   of a set of growing wavefronts can have the same complexity as the
   upper envelope of a set of line segments in the plane.
   The construction is illustrated in Figure~\ref{fig:lower-bound}.

   \begin{figure}[hbt]
      \vspace*{-12pt}
      \centering\includegraphics[width=3in]{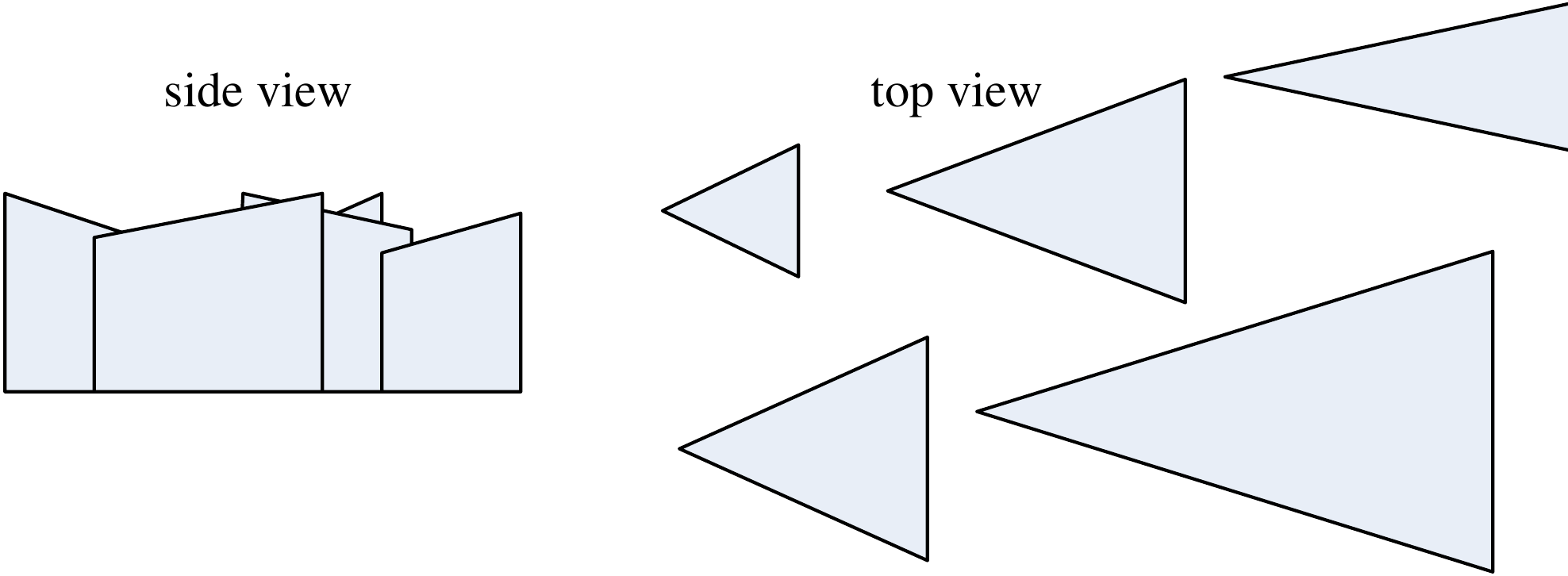}
      \caption{Illustrating 3D skeleton complexity.}
      \label{fig:lower-bound}
   \end{figure}

   The main idea to produce such a cross-section is to set up a sequence
   of triangular prisms sticking up out of a side of the polyhedron.
   Construct the set so that the slopes of their top faces matches those
   of a set of specified line segments and their sides are defined by
   vertical edges corresponding to the segment endpoints.
   Define the wedges in sequence with ever sharper points, so that as
   their wavefronts grow to define the straight skeleton the
   slower-growing wavefronts in the front are overtaken by the faster
   ones in the back, until eventually the complexity of
   the cross-section of the set of growing wavefronts
   matches that of an upper envelope of line segments.
   In this case, we can orient the tip of each wedge so that it will be
   in the visible part of the upper-envelope, which guarantees that the
   cross-section of the latter portion of the set of growing wavefronts
   will have the same complexity as the upper envelope of a set of line
   segments (no matter how the wavefronts are growing in the leading
   portion of this set of growing wavefronts).

   Wiernik and Sharir~\cite{ws-prnds-88} show that the upper envelope
   of a set of line segments \ifFullversion in the plane \fi can have
   $\Omega(n\alpha(n))$ complexity in the worst case.
   Thus, the complexity of the cross-section of the set of growing
   wavefronts in our construction is $\Omega(n\alpha(n))$ in the worst case.
   Our lower bound for the 3D skeleton follows, then, by having such a set
   of growing wavefronts attached to the ``floor'' of a simple polyhedron
   interact with an orthogonal set of similar growing wavefronts attached
   to the ``ceiling'' of a simple polyhedron.
   This is done by making the direction of growing wavefronts much
   longer than their cross-sectional length, which implies that
   as the two sets of wavefronts grow into each other, they produce a
   number of pieces of straight skeleton that is quadratic in the
   complexities of the two sets of wavefronts.
\end{proof}

\comment{

On the other hand, we show the complexity of the straight skeleton is at most
cubic in the complexity of the polyhedron.
\begin{theorem}
   The complexity of the straight skeleton of a general polyhedron is $O(n^3)$.
\end{theorem}

\begin{proof}
   The proof is based on observing the behavior of the event-driven
   algorithm, described in the next section, which computes the skeleton.
   The observation is that during the propagation, every face can be split
   by at most $O(n)$ other faces. More precisely, a face $f$ is actually
   split by edges coinciding to these faces, but, naturally, $f$ cannot be
   split by two different edges adjacent to the same face.
   [[[WHY???]]]
   Therefore, every
   3-cell (of the partition induced by the skeleton) may be adjacent to
   another to another cell by $O(n)$ connected skeletal faces. That gives
   the bound $O(n^2)$ on the total complexity of faces adjacent to a single
   3-cell, for a total complexity of $O(n^3)$.
\end{proof}

}

\subsection{The Algorithm}

\label{SS-gen-algorithm}

\begin{figure}
   \vspace{-6mm}
   \centering
   \subfigure[]{
      \label{fig:InitialTopology}
      \includegraphics[width=1.0In]{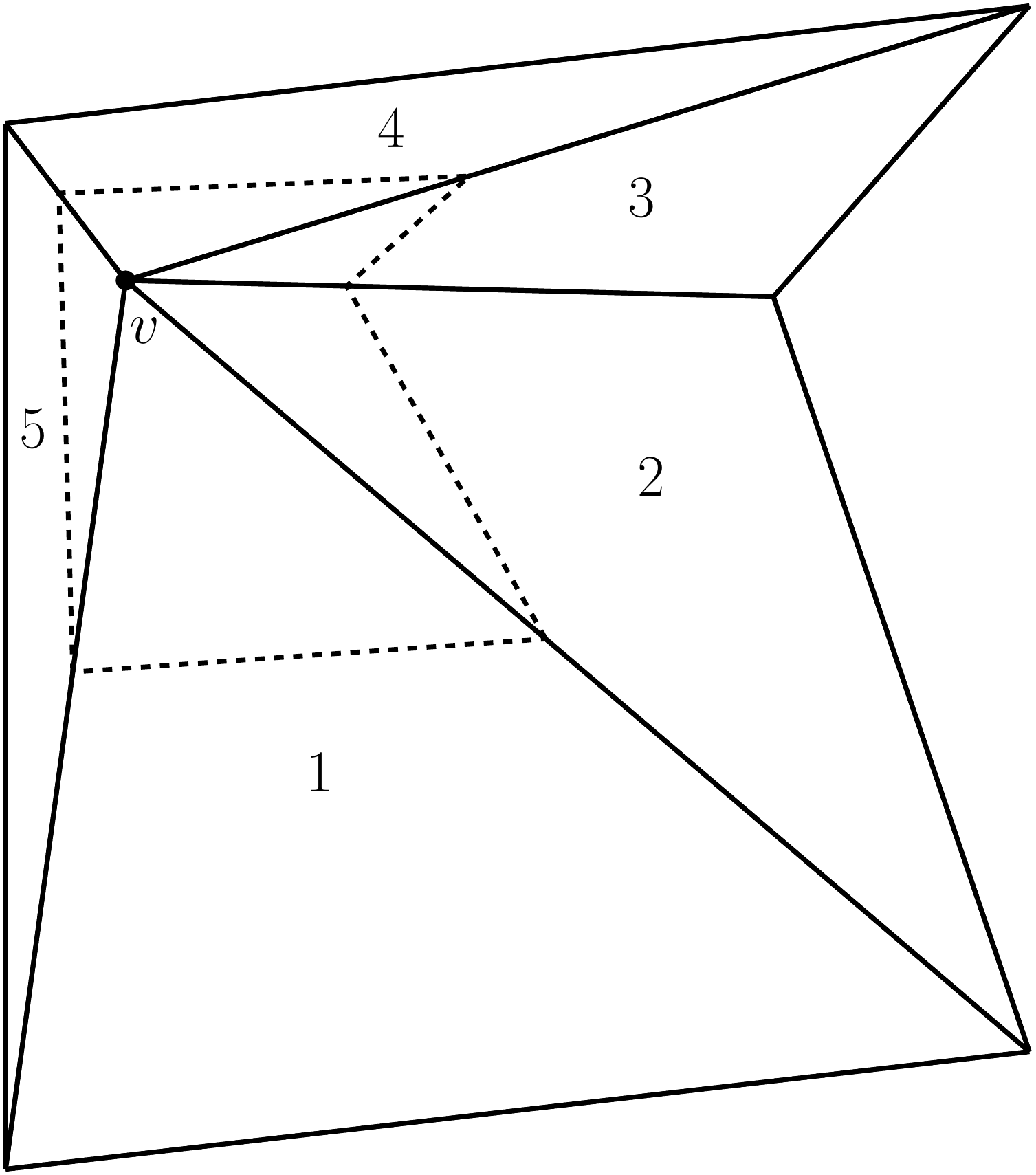}
   }
   \hfill
   \subfigure[]{
      \label{fig:InitialProjection}
      \includegraphics[width=1.0In]{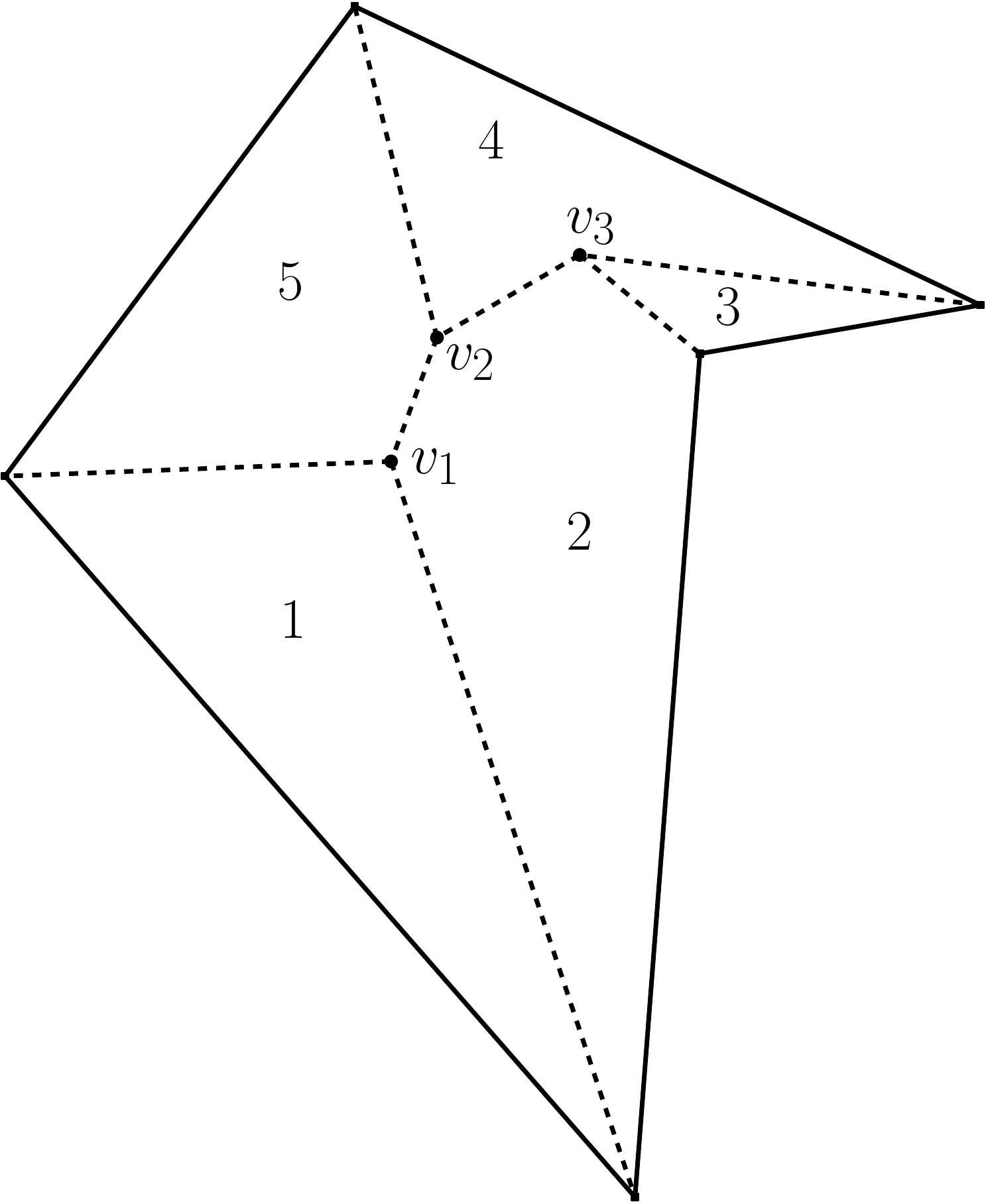}
   }
   \hfill
   \subfigure[]{
      \label{fig:InitialSkeleton}
      \includegraphics[width=1.0In]{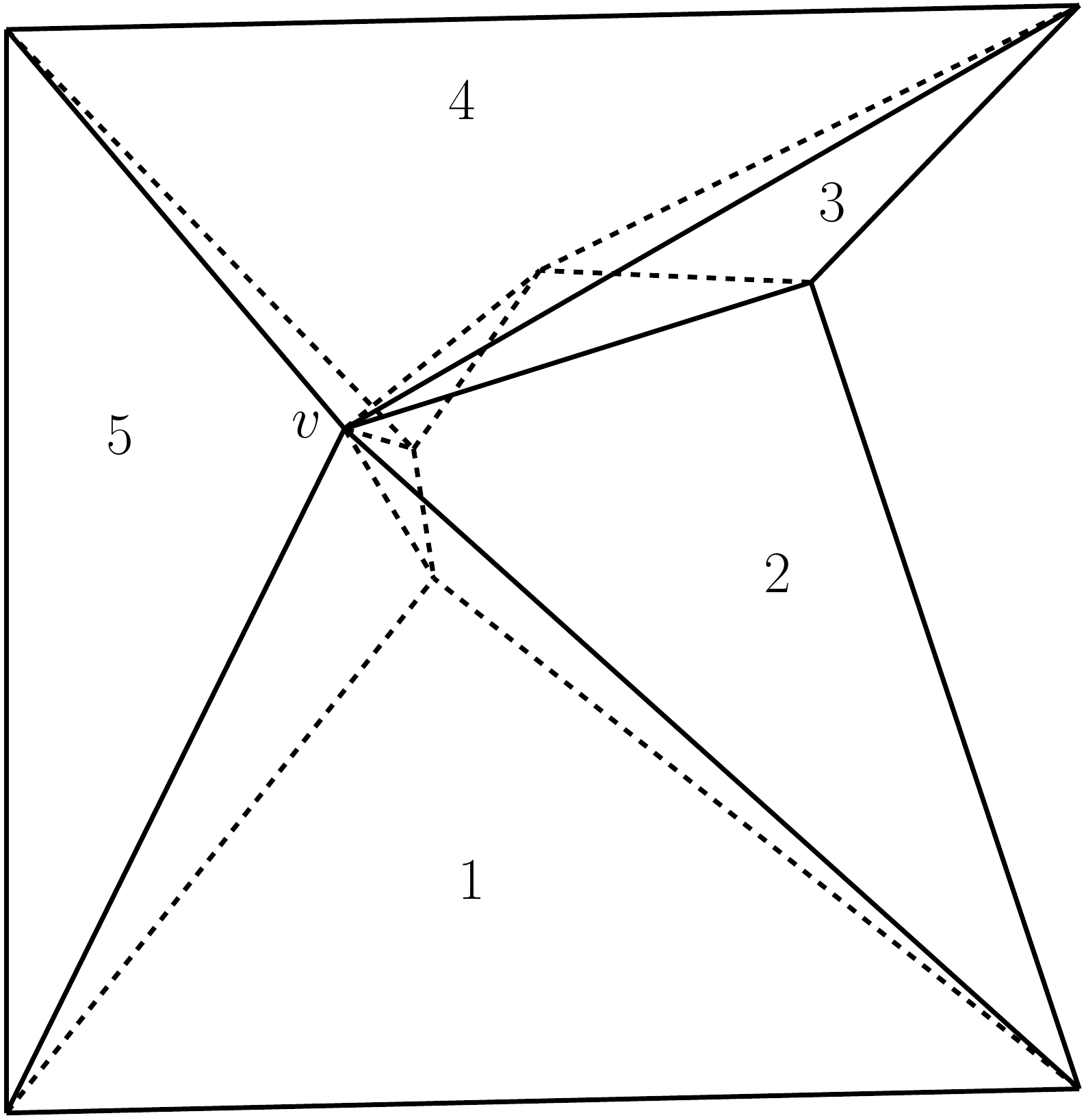}
   }
   \hfill
   \subfigure[]{
      \label{fig:InitialPropagated}
      \includegraphics[width=1.0In]{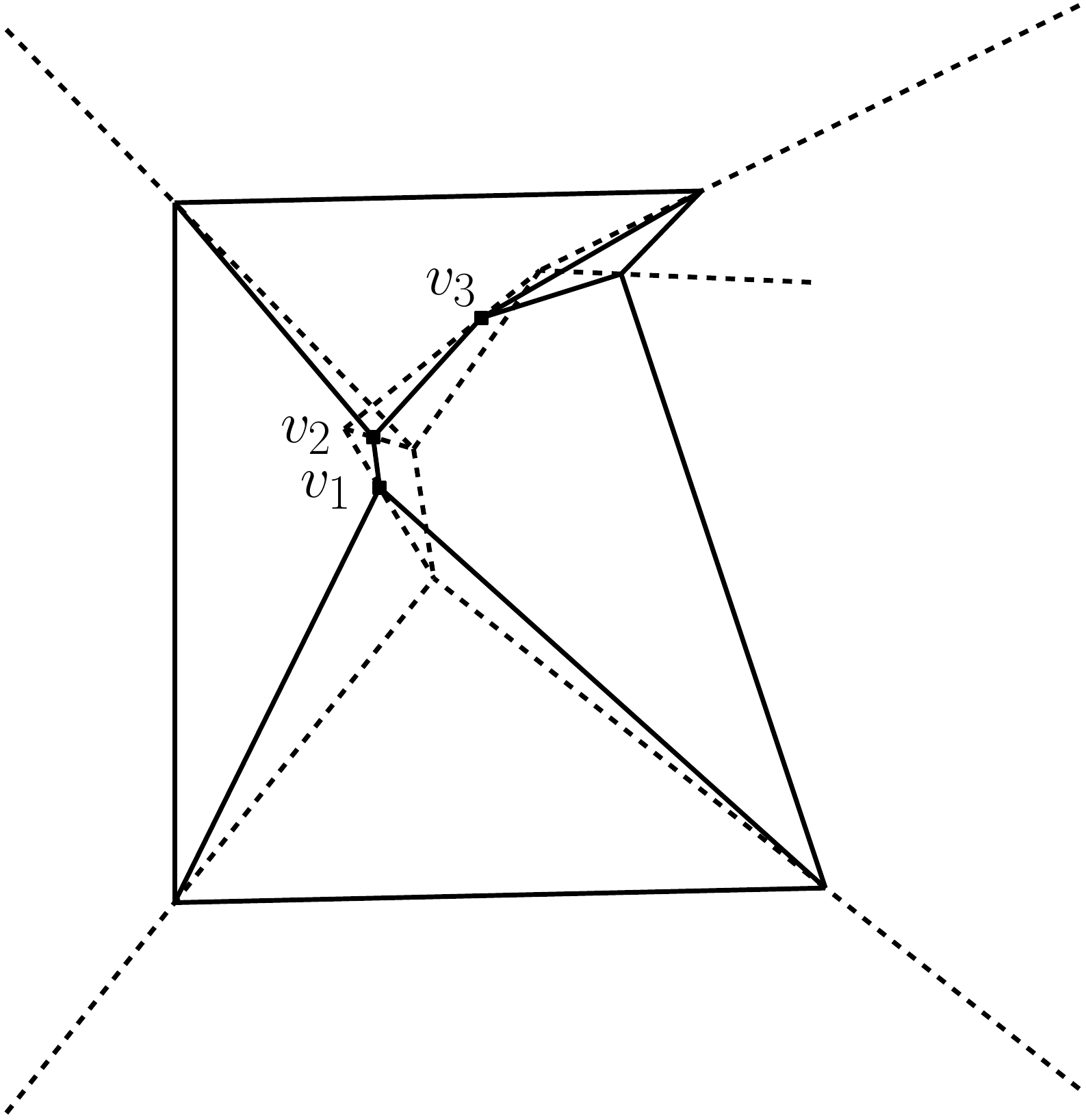}
   }
   \vspace{-3mm}
   \caption{Changing the initial topology of a vertex of degree greater than~3
            (the skeleton is shown in dashed lines):
            (a)~The original polyhedron. Vertex $v$ has degree~5;
            (b)~The cross-section and its weighted straight skeleton.  Vertex
                $v$ becomes three new vertices: $v_1$, $v_2$, and $v_3$;
            (c)~The straight skeleton of the polyhedron. Vertex $v$ spawned
                three skeletal edges;
            (d)~The propagated polyhedron. Vertices $v_1,v_2,v_3$ trace their
                skeletal edges.
           }
   \label{fig:Initial}
\end{figure}

Our algorithm for the general case is an event-based simulation of the propagation of the boundary
of the polyhedron\ifFullversion during the offsetting process\fi.
Events \ifFullversion of the algorithm \fi occur whenever four planes,
supporting faces of the polyhedron, meet at the same point.  At these points
the propagating boundary undergoes topological events.
The algorithm for the general case consists of the following steps:

\begin{enumerate}
\item Collect all possible initial events.
\item While the event queue is not empty:
      \begin{enumerate}
      \item Retrieve the next event and check its validity.
            If the event is not valid, go to Step~2.
      \item Create a vertex at the location of the event and connect to
            it the vertices participating in the event.
      \item Change the topology of the propagating polyhedron according
            to the actions taken in Step~2(c). Set the location of the
            event to the newly-created vertices.
      \item Create new events for newly-created vertices, edges, and faces
            and their neighbors, if needed.
      \end{enumerate}
\end{enumerate}

We next describe the different events and how each type is
dealt with.  The procedure always terminates since the number of all
possible events is bounded from above by the number of combinations of
four propagating faces.

\ifFullversion
\subsubsection{Initial Topology}
\else
\paragraph{Initial Topology.}
\fi
At the start of the propagation, we need to split each vertex of degree
greater than~3 into several vertices of degree~3 (see
Figure~\ref{fig:Initial}).
This is the ambiguous situation discussed earlier; it can have
several valid solutions. Our approach is based on cutting the faces
surrounding the vertex with one or more planes (any
cutting plane intersecting all faces and parallel to none suffices),
and finding the weighted straight skeleton of the intersection of these
faces with the cutting plane, with the weights determined by the dihedral
angles of these faces with the cutting plane, after an infinitesimally-small
propagation. The topology of this two-dimensional straight skeleton tells us the connectivity to use subsequently, and always yields a unique valid solution.
We establish this method for all types of vertices:
\begin{itemize}
\item \textbf{Convex vertices and spikes.}
      In a convex vertex, all of the edges are on the negative side of its
      osculating plane. Edges adjacent to convex vertices do not have to be
      convex. A spike is the opposite of a convex vertex, as all edges are
      on the positive side of the osculating planes. The topologies of both
      types can be determined by sectioning, propagating, and finding the
      weighted straight skeleton.
\item \textbf{Saddle Vertices.}
      In saddle vertices, some of the edges lie on the negative side of the
      osculating plane, denoted ``up'' edges, and some on the positive side, 
     denoted ``down'' edges. We use two sectioning planes. First, we section 
     the edges in the negative side of the plane. Then, we construct their 
     straight skeleton. The section is not a closed polygon, as the 
     intersections of the faces that lead from a ``down'' edge to an ``up'' 
     edge, and vice versa, are infinite rays. Every pair of such infinite 
     rays creates a
     wedge\ifFullversion (see Figure~\ref{fig:SaddleDownSection})\fi. When 
     computing the weighted straight skeleton of this planar shape, one 
     skeletal vertex will be adjacent to an infinite ray for each such wedge. 
     We call these vertices ``portals.''
     Next, we take each  ``up''-streak (i.e., a set of ``up'' edges that are
     consecutive in a conuterclockwise order around the saddle vertex) and
     section it by a second plane above the osculating plane. We get a
     segment-chain, beginning and ending in infinite rays.
\ifFullversion%
     (See Figure~\ref{fig:Saddles}.)
     \begin{figure}
        \centering
        \subfigure[Saddle Vertex]{
        \label{fig:SaddleVertex}
        \includegraphics[width=1.8In]{saddle-original}}
        \hfill
        \subfigure[Weighted Straight Skeleton of the ``down'' section]{
        \label{fig:SaddleDownSection}
        \includegraphics[width=1.8In]{saddle-down-section}}
        \caption{Treating a saddle vertex. The vertex~0 is the saddle, and vertices~1-7, marked as ``up'' or ``down'' vertices are shown in (a). In (b), the section of the down-edges is shown, and the portal vertices are marked by circles. These portals are connected to the respective up vertices (since there are three up streaks, each containing only one vertex).}
        \label{fig:Saddles}
     \end{figure}
\fi%
     We calculate the straight skeleton of the segment-chain, resulting in a
     single vertex of degree~2.\ifFullversion\footnote{In practice, we extend the last two
     segments by a large-enough distance, and connect them, so this vertex is
     the last edge event of this faraway edge of a closed polygon. We do a
     similar process in the ``down'' step in order to avoid computing the
     straight skeleton of unbounded shapes.}\fi{}
     Every ``up''-streak corresponds
     to a ``down''-wedge, and we connect the last ``up'' vertex with its
     corresponding ``portal'' vertex by an edge. Thus, we get a new
     fully-connected topology for saddle vertices.
\end{itemize}

\ifFullversion
\subsubsection{Collecting Events}
\else
\paragraph{Collecting Events.}
\fi
\ifFullversion
The heart of the computation of events is finding the intersection of four
planes. Let $P_i:a_ix+b_iy+c_iz+d_i=0,\:0\leq i \leq 4$, be the plane
equations, and let $N_i=\sqrt{a^2_i+b^2_i+c^2_i}$. By solving the linear
system 
\begin{equation}
\left(
  \begin{array}{cccc}
     a_1 & b_1 & c_1 & N_1 \\
     a_2 & b_2 & c_2 & N_2 \\
     a_3 & b_3 & c_3 & N_3 \\
     a_4 & b_4 & c_4 & N_4 \\
  \end{array} \right)
\left(
  \begin{array}{c}
     x \\ y \\ z \\ t
  \end{array} \right) = 
\left(
  \begin{array}{c}
     -d_1 \\ -d_2 \\ -d_3 \\ -d_4
  \end{array} \right), 
\end{equation}
we obtain the coordinates and time of the event. In general position, this
system has a unique solution. However, the time can be less than zero (or
less than the current time if the event is computed within Step~2 of the
algorithm), in which case the event is not valid. In addition, the
locations of the events must be within the area the participating edges
sweep, i.e., bounded between the edge and the trisectors (the rays
traced by vertices during propagation) of its endpoints.

\noindent \textbf{Edge events.}
There are three types of edge-event outcomes (see
Figure~\ref{fig:EdgeEvents}):
\begin{figure}
   \subfigure[Before edge event]{
   \label{fig:EdgeEventBefore}
   \includegraphics[width=1.2In]{edge_event_before}}
   \hfill
   \subfigure[After edge event]{
   \label{fig:EdgeEventAfter}
   \includegraphics[width=1In]{edge_event_after}}
   \hfill
   \subfigure[Before an edge event which closes a face]{
   \label{fig:FaceEventBefore}
   \includegraphics[width=1.6In]{face_event_before}}
   \hfill
   \subfigure[After closing a face]{
   \label{fig:FaceEventAfter}
   \includegraphics[width=1.2In]{face_event_after}}
   \caption{The situation before and after edge events: in (a,b), edge $e$ vanishes and edge $e'$ is created, and in (c,d) face~4 vanishes.}
   \label{fig:EdgeEvents}
\end{figure}
\begin{itemize}
\item An edge vanishes as its four surrounding planes meet. In this event the
      two planes that were not adjacent become adjacent, sharing a
      newly-created edge. Therefore, an edge vanishes, and a new edge appears.
\item A face vanishes when its last three edges vanish simultaneously. The
      three planes adjacent to the vanishing face now share a vertex. No new
      edges are created.
\item A simplex vanishes when its four faces (and four edges) vanish
      simultaneously. This is the last event for each connected component in
      the propagation of a polyhedron.\footnote{A convex polyhedron
      propagates into a single connected component, as it does not split.}
\end{itemize}

\noindent \textbf{Hole events.}
The reflex vertex establishes a new triangular hole within the face it hits
(see Figure~\ref{fig:HoleEventBeforeAfter}).
\begin{figure}
   \subfigure[Hole event]{
   \label{fig:HoleEventBeforeAfter}
   \includegraphics[width=1.8In]{hole_event_before_after}}
   \subfigure[Split event]{
   \label{fig:SplitEventBeforeAfter}
   \includegraphics[width=1.8In]{split_event_before_after}}
   \subfigure[Edge-split event]{
   \label{fig:EdgeSplitEventBeforeAfter}
   \includegraphics[width=1.8In]{edge_split_event_before_after}}
   \hfill
   \subfigure[Vertex event]{
   \label{fig:VertexEventBeforeAfter}
   \includegraphics[width=1.8In]{vertex_event_before_after}}
   \caption{Topology changes in hole, edge and, edge-split events. The dashed lines denote the topology before the event (in (d), dotted lines show the second split event done by the same edge. The endpoints of the propagating split edge here are marked by squares. The location of the immediate vertex event is marked by a circle), and the solid lines denote the topology after the change.}
   \label{fig:EventBeforeAfter}
\end{figure}
The face is then adjacent to the three faces defining the reflex vertex,
sharing the edges of the hole with them. The three vertices of the hole are
new ridges in the hit face.

\noindent \textbf{Split events.}
The ridge vertex splits through the edge and creates a new ridge in the
other face adjacent to the split edge (see
Figure~\ref{fig:SplitEventBeforeAfter}).  The edge splits into two parts.

\noindent \textbf{Edge-split events.}
New Ridges are created, in the same manner as in the split event, on all
four faces adjacent to the edges (see
Figure~\ref{fig:EdgeSplitEventBeforeAfter}).

\noindent \textbf{Vertex events.}
Two ridges sharing a reflex edge unite. the face which both ridges share is
split (see Figure~\ref{fig:VertexEventBeforeAfter}).
\else
In the full version of the paper we describe how events are
collected, classified as valid or invalid, and handled by the algorithm.
In a nutshell, each processed event arises from interactions of features of the wavefront, and gives rise to potential future events, 
all of which may be specified by the interaction of four planes.
However, a potential event may be
found invalid (not corresponding to an actual event in the propagation of the wavefront) already when it is created (since its time stamp is less than
that of the current event, or because its geometric location is outside
its ``region of influence''), or later when it is fetched for processing
(if another already-processed event has annulled it).  Each valid event
results in the creation of features of the skeleton, and in a topological 
change in the structure of the propagating polyhedron. This part of the algorithm requires a careful case analysis but is conceptually straightforward.
\fi

\ifFullversion
\subsubsection{Handling Events}
\else
\paragraph{Handling Events.}
\fi
Propagating vertices are defined as the intersection of propagating planes.
Such a vertex is uniquely defined by exactly three planes, which also define
the three propagating edges adjacent to the vertex. (When an event creates a
vertex of degree greater than~3, we handle it as as in the initial
topology---see above.)
The topology of the polyhedron remains unchanged during the propagation
between events. Figure~\ref{fig:Events} depicts all possible events:
\begin{figure}
   \vspace{-9mm}
   \centering
   \subfigure[Edge]{
      \label{fig:EdgeEvent}
      \includegraphics[width=1.1In]{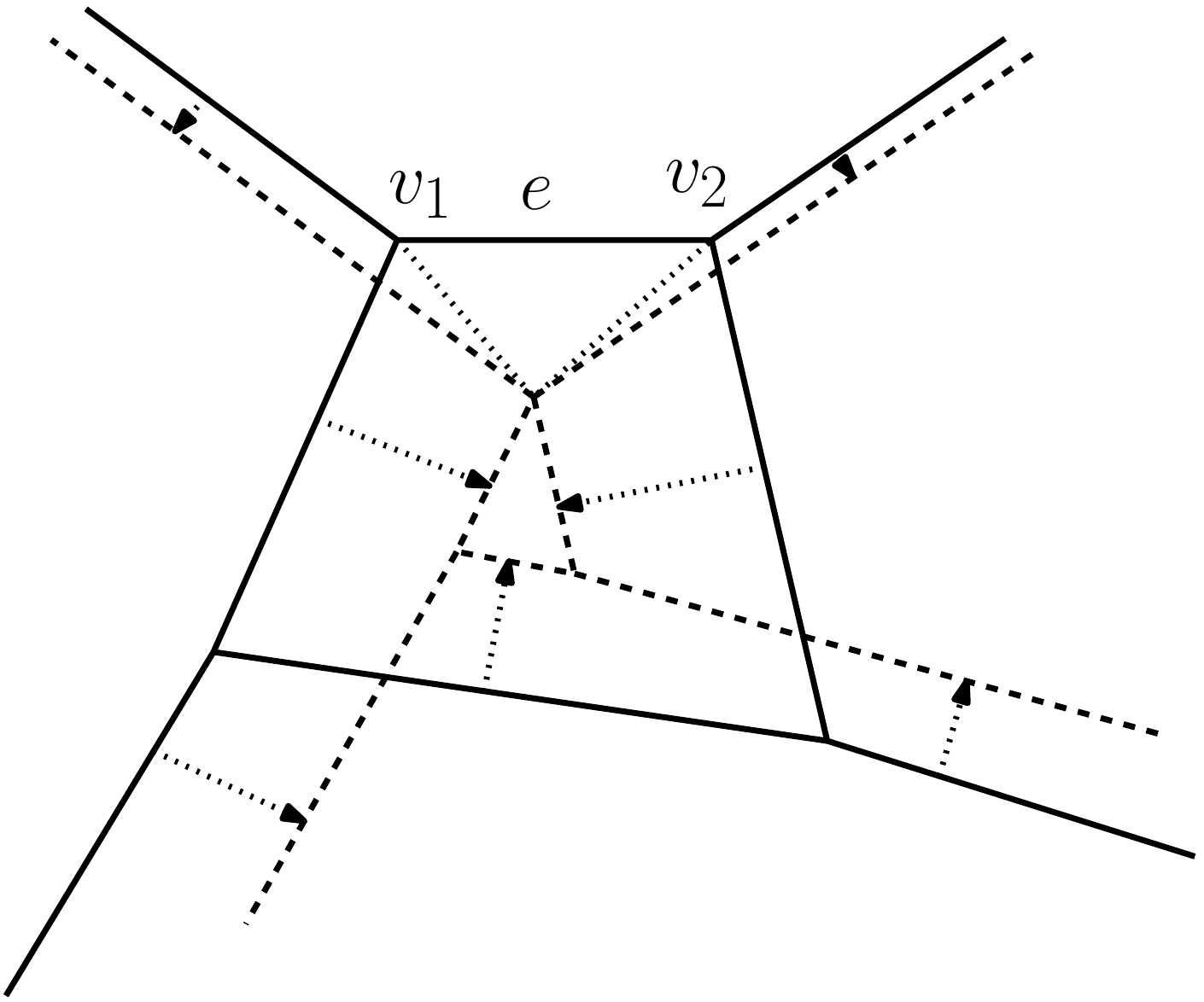}
   }
   \hfill
   \subfigure[Hole]{
      \label{fig:HoleEvent}
      \includegraphics[width=1.1In]{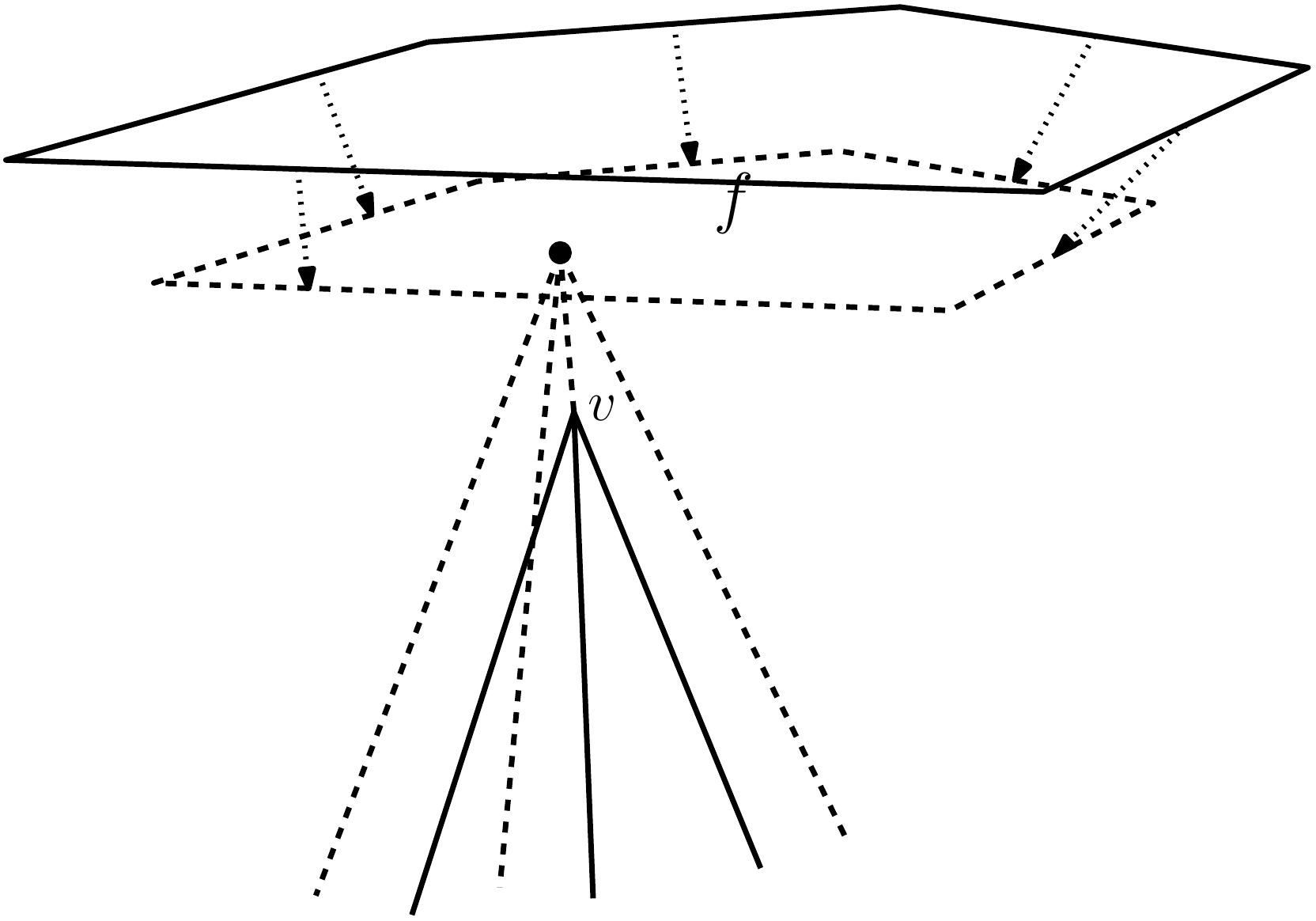}
   }
   \hfill
   \subfigure[Split]{
      \label{fig:SplitEvent}
      \includegraphics[width=0.65In]{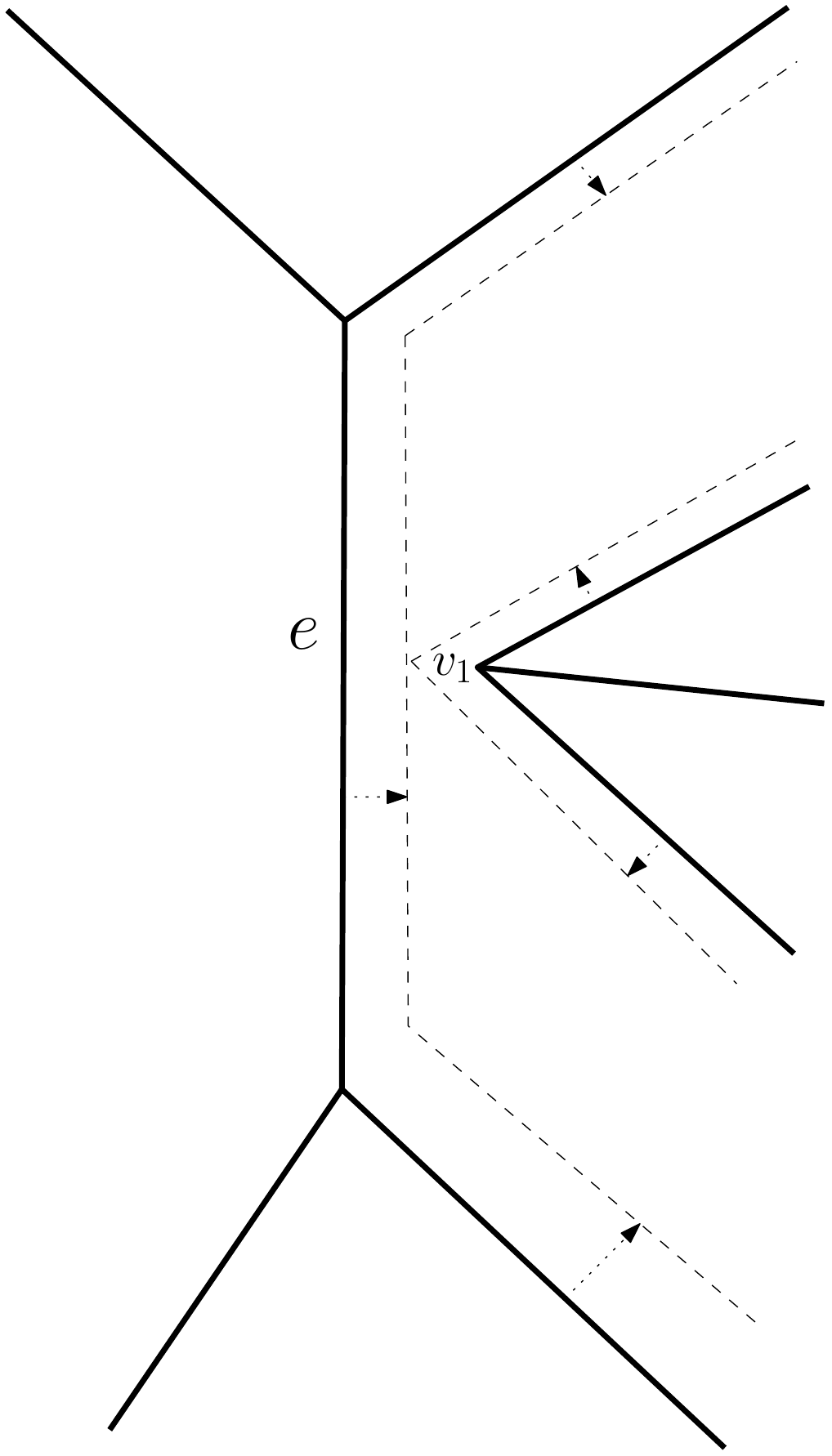}
   }
   \hfill
   \subfigure[Edge-Split]{
      \label{fig:EdgeSplitEvent}
      \includegraphics[width=1.1In]{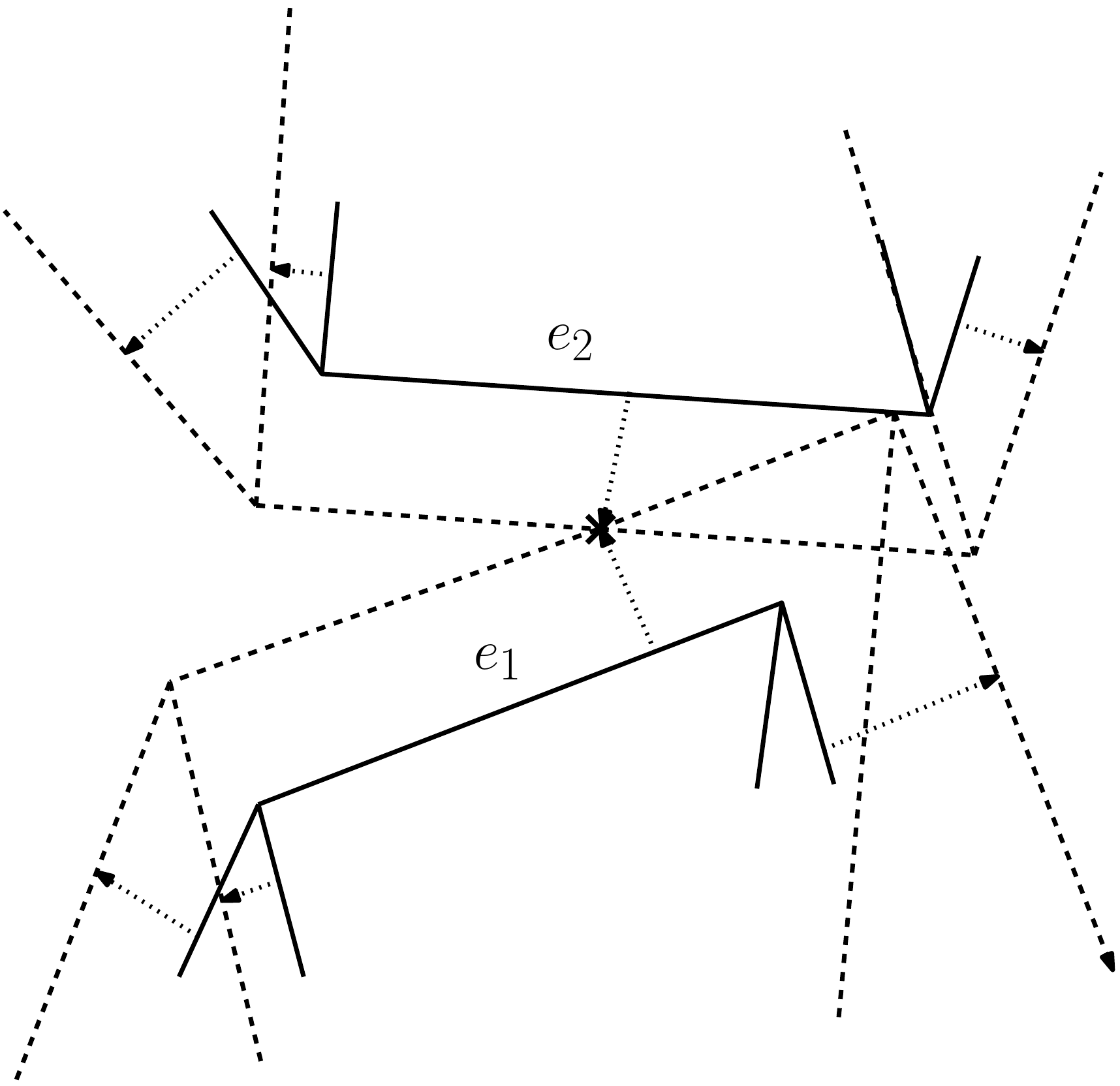}
   }
   \hfill
   \subfigure[Vertex]{
      \label{fig:VertexEvent}
      \includegraphics[width=1.0In]{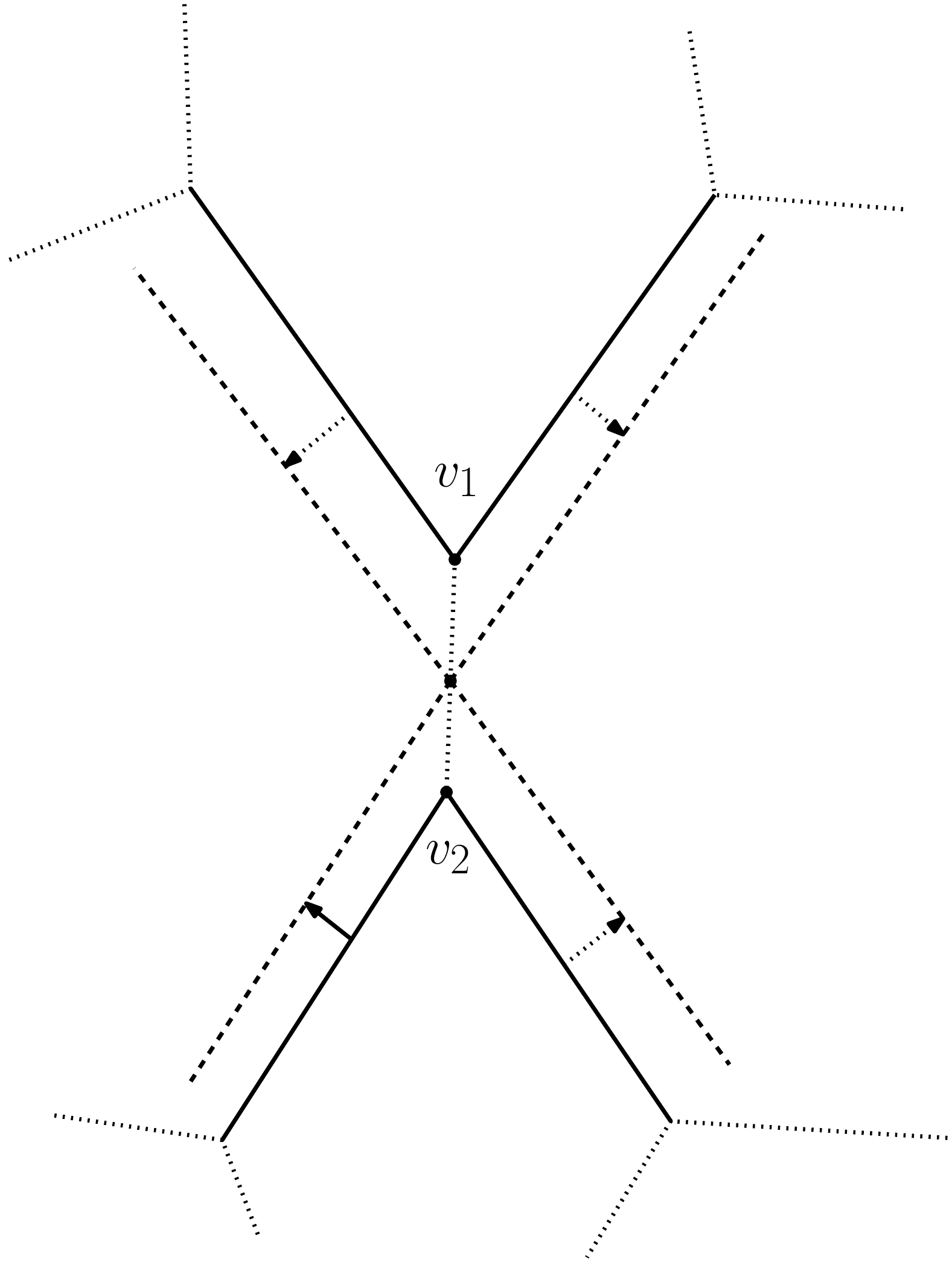}
   }
   \caption{The five types of events. The solid lines are the original edges,
            and the dashed lines are their locations after the propagation.
            The dotted arrows show the progression of these edges, up to the
            time of the event.}
\label{fig:Events}
\end{figure}
\begin{enumerate}
\item \textbf{Edge Event.}
      An edge vanishes as its two endpoints meet, at the meeting point of the four planes around the edge.
\item \textbf{Hole Event.}
      A reflex vertex (adjacent to three reflex edges, called a ``spike'')
      runs into a face. The three planes adjacent to this vertex meet the
      plane of the face. After the event, the spike meets the face in a small triangle.
\item \textbf{Split Event.}
      A ridge vertex (adjacent to one or two reflex edges) runs into an
      opposite edge. The faces adjacent to the ridge meet the face adjacent
      to the twin of the split edge. This creates a vertex of degree greater than 3, handled as in the initial topology.
\item \textbf{Edge-Split event.}
      Two reflex edges cross each other. Every edge is adjacent to two planes.
\item \textbf{Vertex event.}
      Two ridges sharing a common reflex edge meet. This is a special case of
      the edge event, as it is the meeting of the endpoints of the reflex
      edge, but it has different effects, and so it is considered a different
      event. Vertex events occur when a reflex edge runs twice into a face,
      and the two endpoints of this edge meet.
\end{enumerate}

A convex polyhedron induces only edge events during propagation, and reduces
to a single tetrahedron before vanishing at the simultaneous edge
events of the last four edges. A general (nonconvex) polyhedron may split into
several connected components, which will be reduced into tetrahedra and
similarly vanish. All these events are meeting points of four planes, and
other types of events are not accounted for, as they do not occur in general
position (e.g., two reflex vertices running into each other), which are
meeting points of more than four planes at a location.
\ifFullversion
\else
Note that the propagation of the boundary is ``memoryless'' in the sense
that handling an event does not depend on the history of propagation.
Therefore, degenerate events are treated exactly the same as initial
vertices of degree greater than~3.
\fi

\ifFullversion
\subsubsection{Data Structures}
\else
\paragraph{Data Structures.}
\fi
We use an event queue which holds all possible events sorted by time, and a
set of propagating polyhedra, initialized to the input polyhedron (or
polyhedra), after the initialization of topology.
\ifFullversion
The propagating polyhedra
(the generalization of the SLAV structure in two dimensions) contain only
topological information and locations of vertices at their creation; no
intermediate offset polyhedra are constructed.\footnote{
   Technically, it is easy to construct such offset polyhedra by computing
   the intersections of adjacent planes, offset by the given time.
} When an event occurs, we connect the vertices that participate in the
event with skeletal edges, with one end at the vertices in their former
location, and the other at the location of the current event, thus
connecting the locations of last and current events. The faces and 3-cells
are formed in a post-processing step, although they can easily be formed
on-line.
\else
The used structure is a generalization of the SLAV structure in two
dimensions.  We provide the details in the full version of the paper.
\fi

\ifFullversion
\subsubsection{Degeneracies}
A degeneracy is usually created where the polyhedron has some symmetry. The
most common degeneracy involves parallel features (mostly faces, like in a
box) collapsing one to the other, or concurrent propagation of more than
three planes (such as in a vertex which cross-section is a regular polygon).
These can be easily handled with a slight perturbation and the usage of
zero-length edges to get correct results.  Uncommon (but existing)
degeneracies are reflex events that involve more than four planes---such as
events that involve two ``spikes'' or two ridges not sharing a common reflex
edge (i.e., unlike our vertex event). All of these events can be treated as
a composition of the five basic events (e.g., two meeting ridges can be
treated as a split event following edge events), but they must be identified
precisely beforehand, as regular perturbation can result, since in the
two-dimensional case, in wrong results.

\fi

\ifFullversion
\subsubsection{Running Time}
\else
\paragraph{Running Time.}
\fi
Let $n$ be the total complexity of the polyhedron, and $k$ the number
of events processed by the algorithm\ifFullversion; $k$ is also the
complexity of the straight skeleton\fi.
Let $r$ be the number of reflex vertices (or edges) of the polyhedron.
\ifFullversion
Naturally, $k=\Omega(n)$ and $r=O(n)$.
\fi
In order to collect all the initial events, we iterate over all
vertices, faces, and edges of the input polyhedron.  Edge events require
only looking at each edge's neighborhood, which can be done in $O(n)$
time.
However, finding all hole events requires considering all pairs of a
reflex vertex and a face.  This takes $O(rn)$ time.
Computing a split event is bounded within the edges of the common face,
but this can still take $O(rn)$ time, and computing Edge-Split
events takes $O(r^2)$ time.

In the course of the algorithm, we compute and process future events.  For a convex polyhedron, only edge events are created, and are easily computed locally in $O(1)$ time per event. However, for a
general polyhedron, every edge might be split by any ridge and stabbed by
any spike. In addition, new spikes and ridges can be created when events
are processed, and they have to be tested against all other vertices, edges,
and faces of their propagating connected component. Since $O(1)$ vertices and
edges are created in every event, every event can take $O(n)$ time to handle.
(The time needed to perform queue operations per a single event,
$O(\log{n})$, is comparatively negligible.)
The total time needed for processing the events is, thus, $O(kn)$.
This is also the total running time of the algorithm.

\begin{figure}
   \centering
   \ifTitlepage {\LARGE\bf Appendix} \\[1.5in] \fi
   \begin{tabular}{cccccccc}
      \includegraphics[width=0.75in]{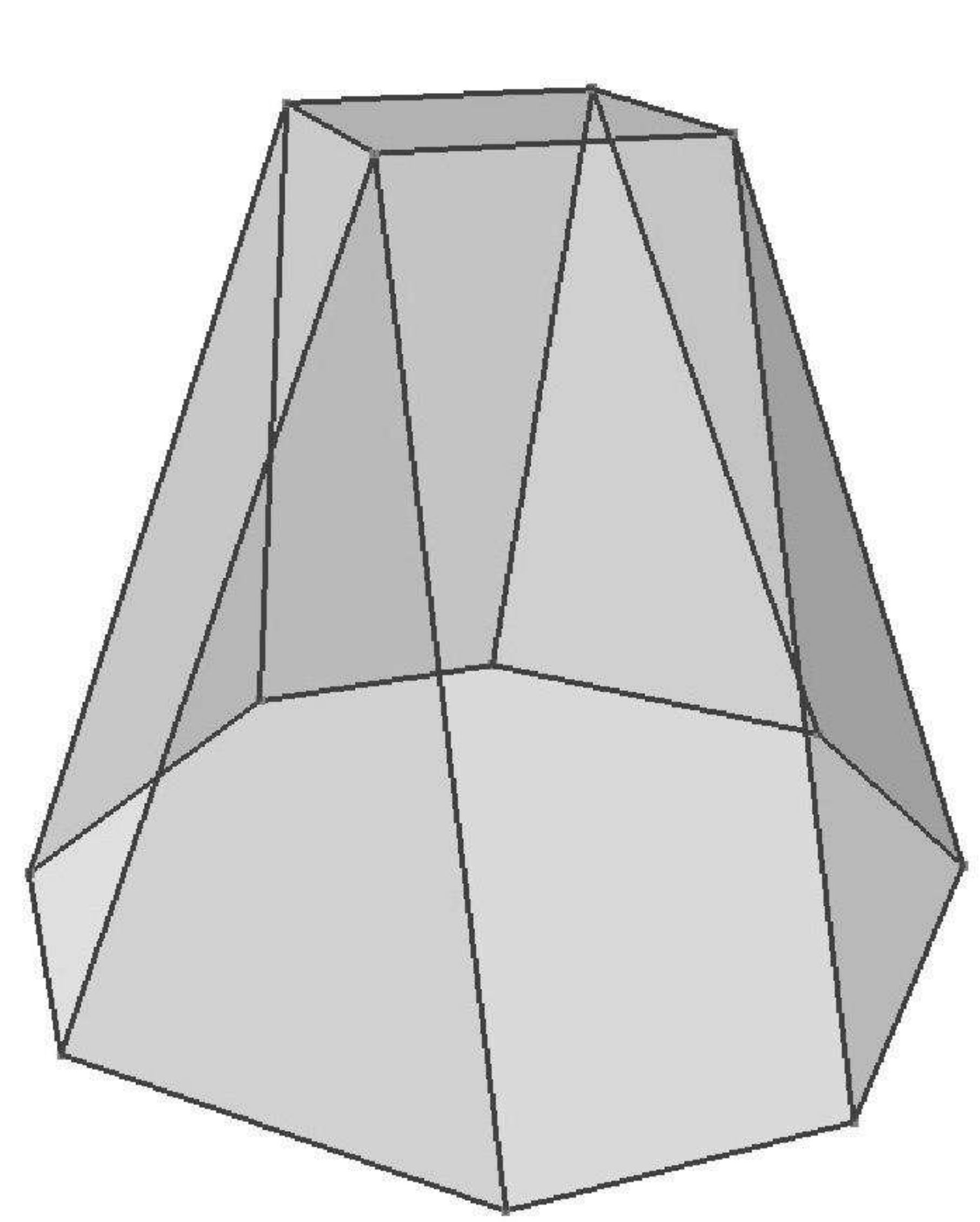} &
         \includegraphics[width=0.80in]{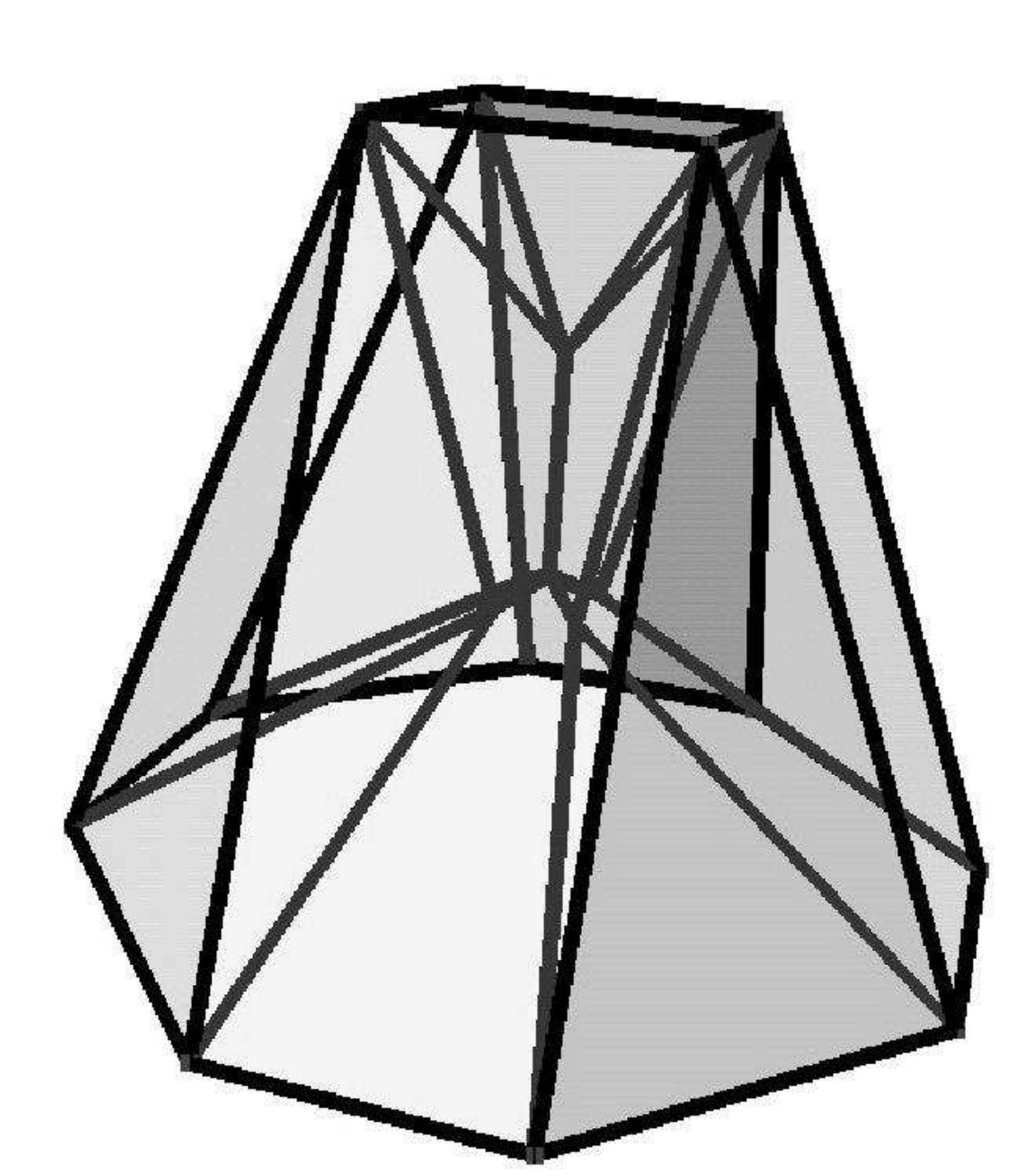} & ~~~~~~~ &
         \includegraphics[width=0.80in]{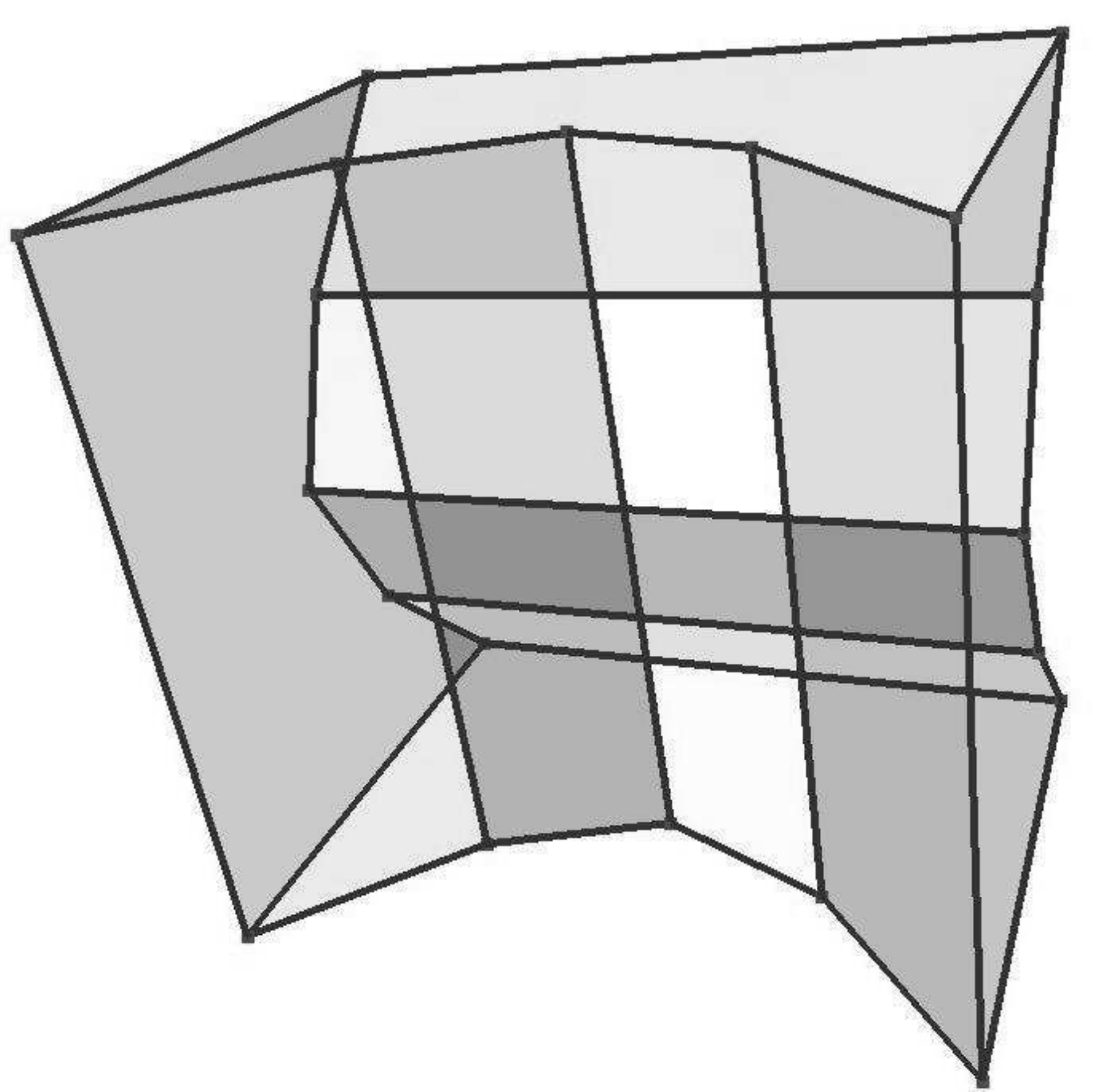} &
         \includegraphics[width=0.80in]{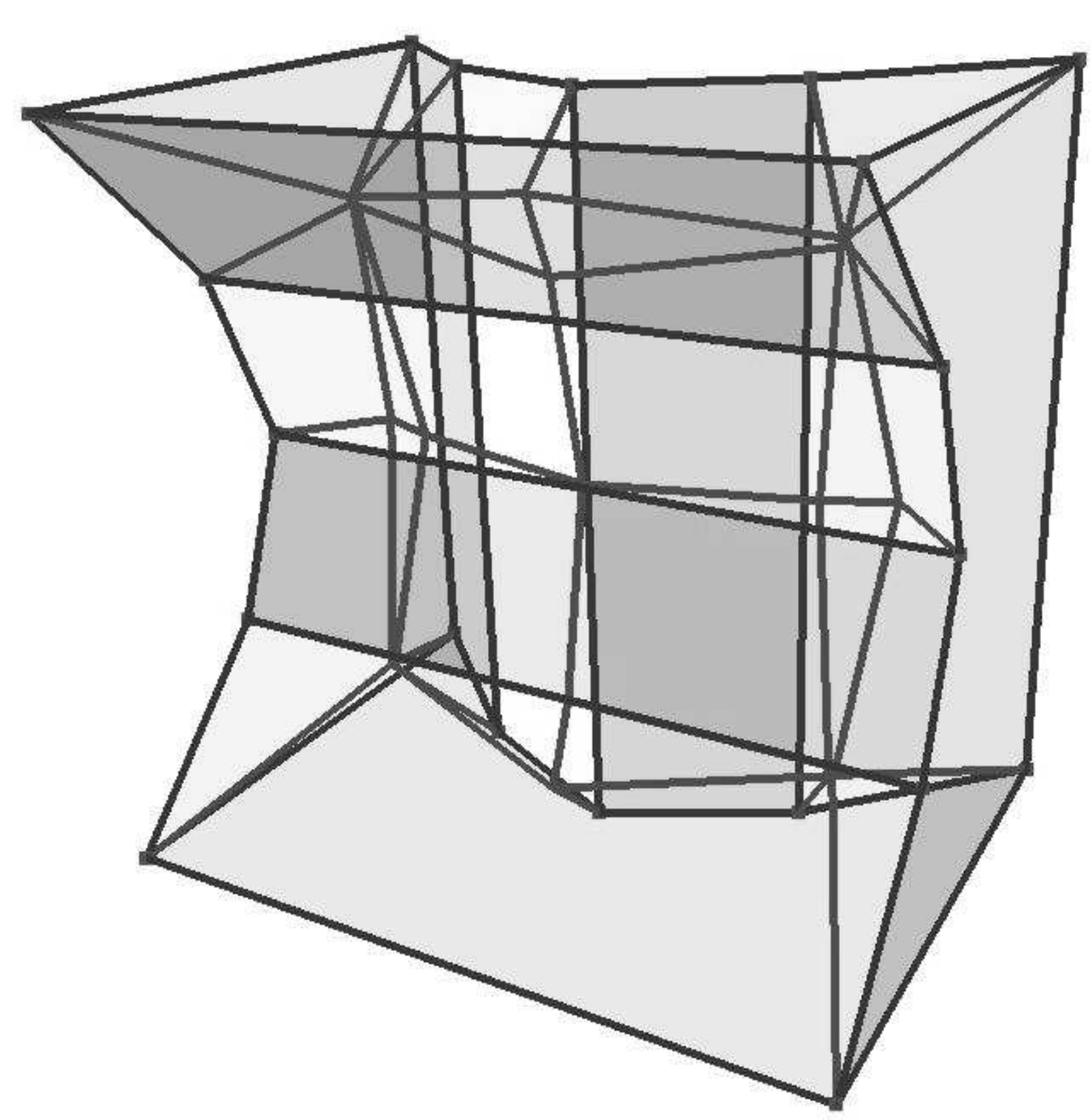} & ~~~~~~~ &
         \includegraphics[width=0.75in]{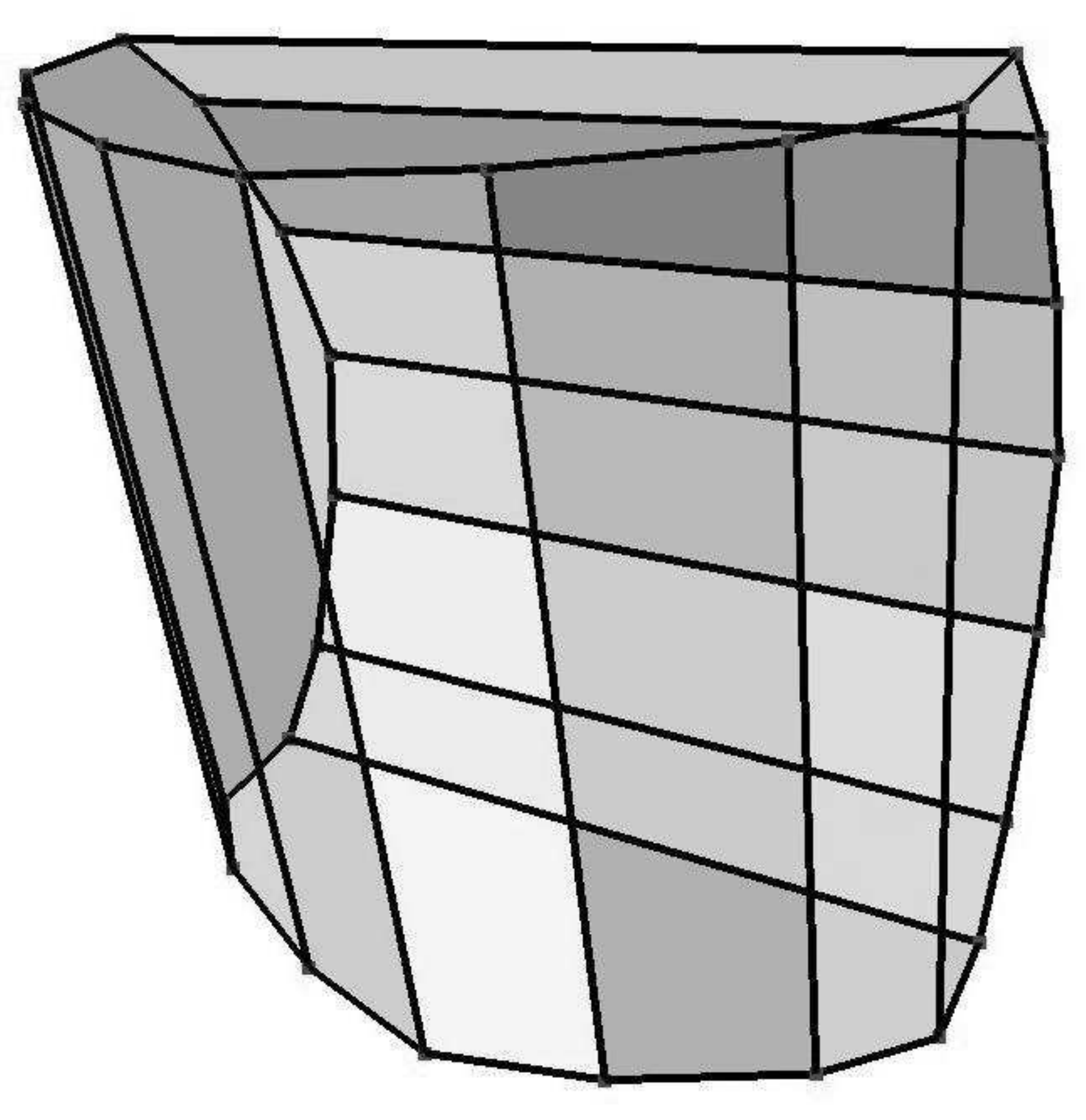} &
         \includegraphics[width=0.80in]{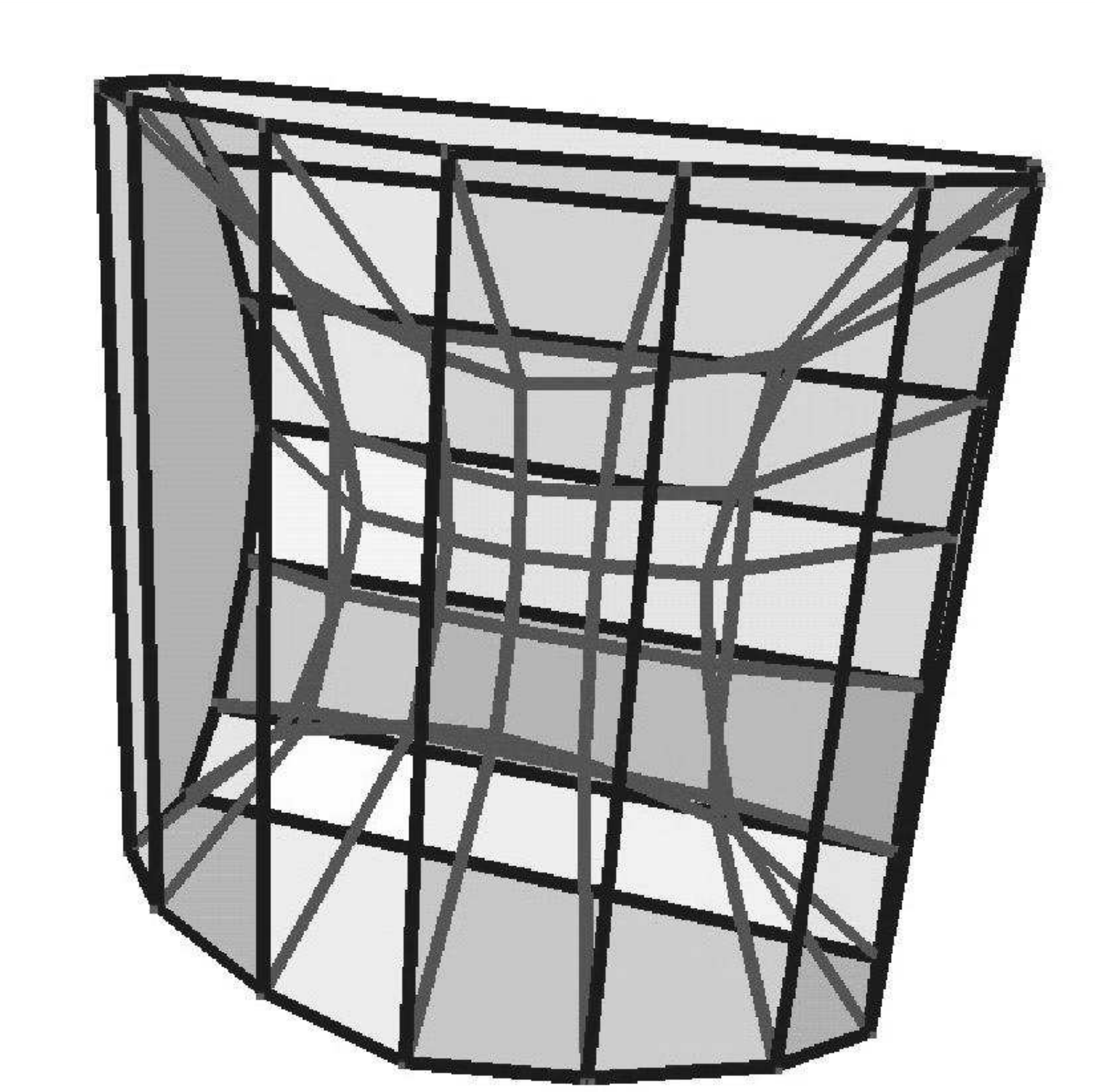} \\
      \multicolumn{2}{c}{(a)} & & \multicolumn{2}{c}{(b)} & &
         \multicolumn{2}{c}{(c)} \medskip \\
      \includegraphics[width=0.75in]{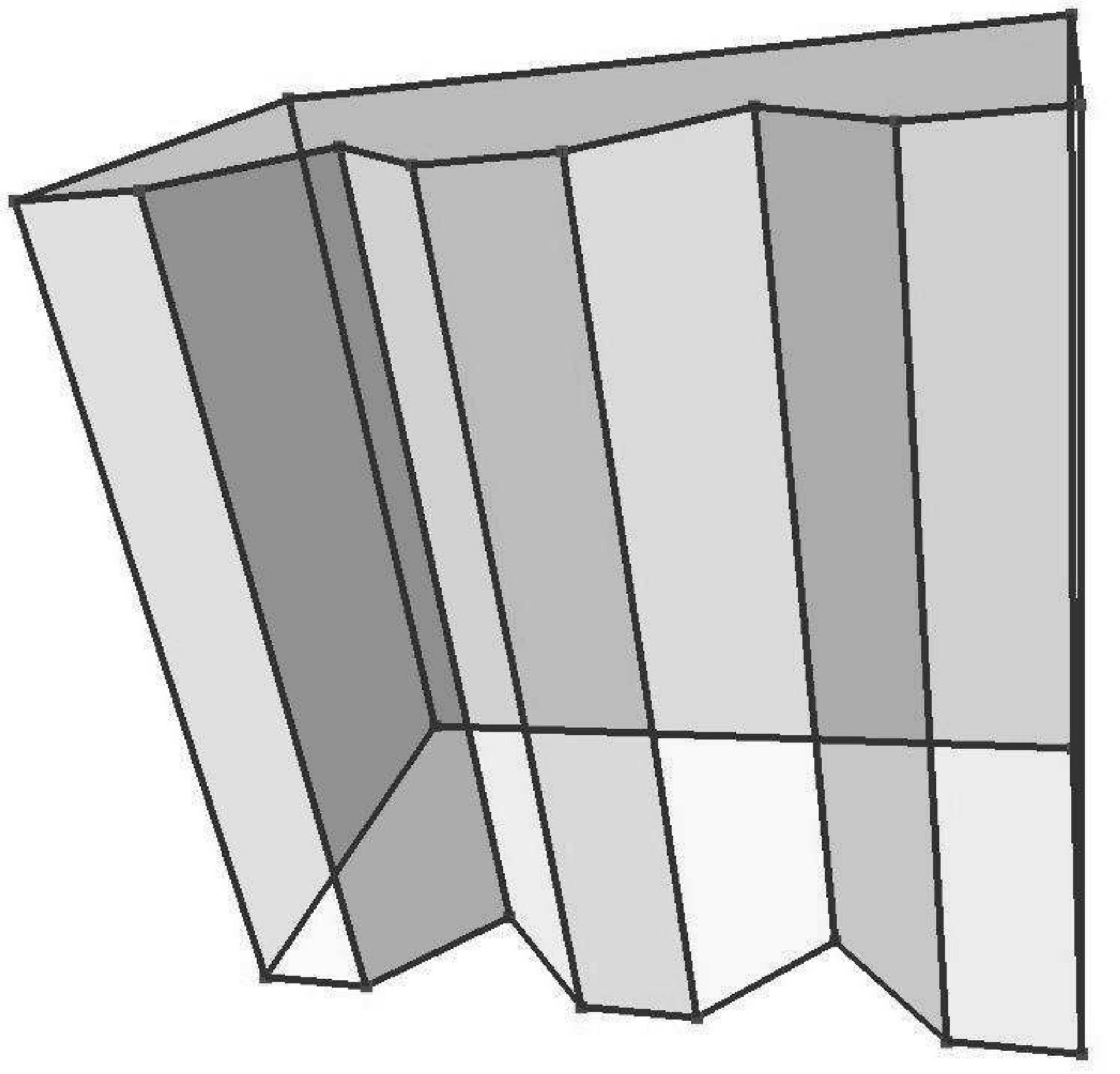} &
         \includegraphics[width=0.75in]{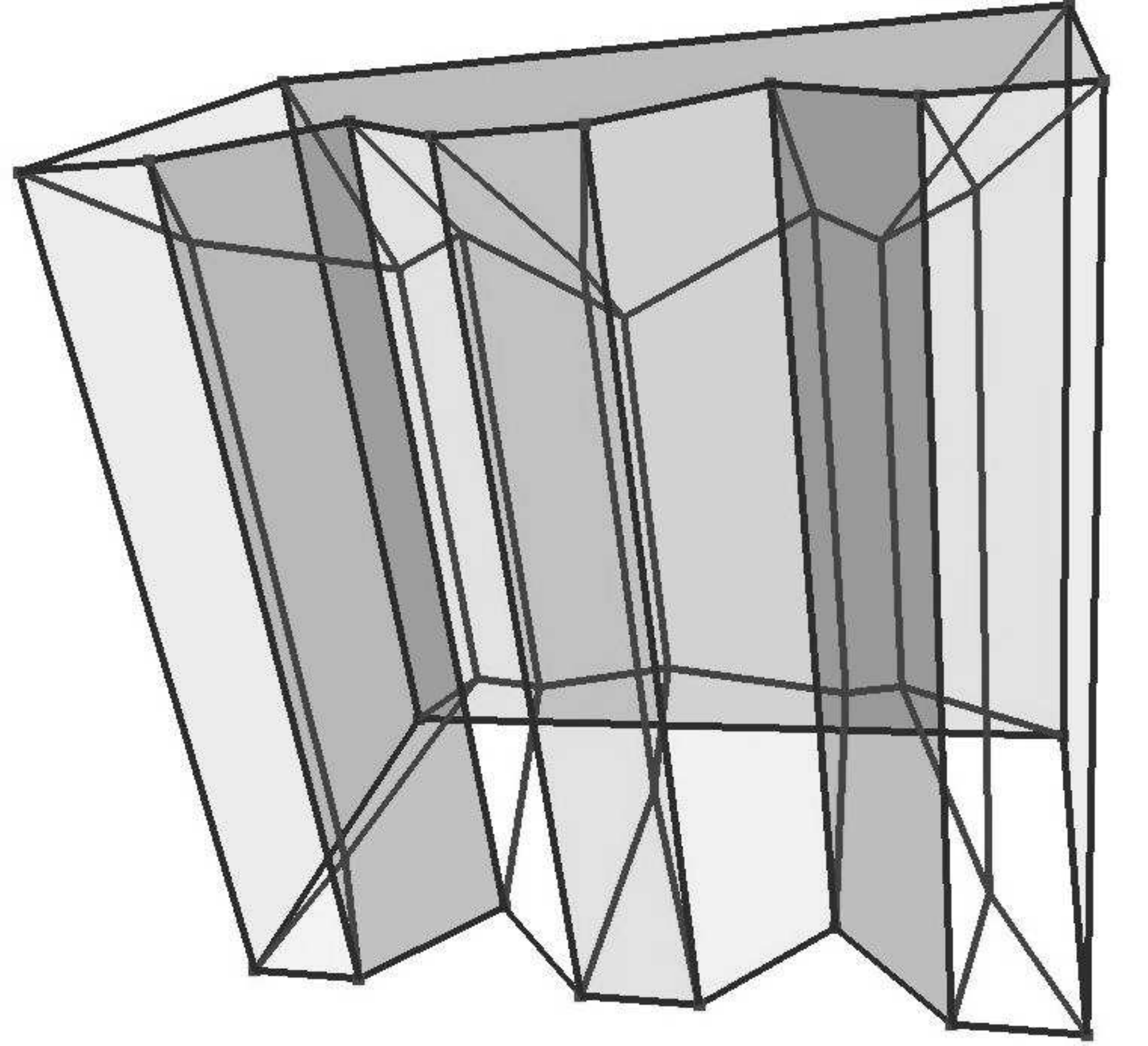} & &
         \includegraphics[width=1.00in]{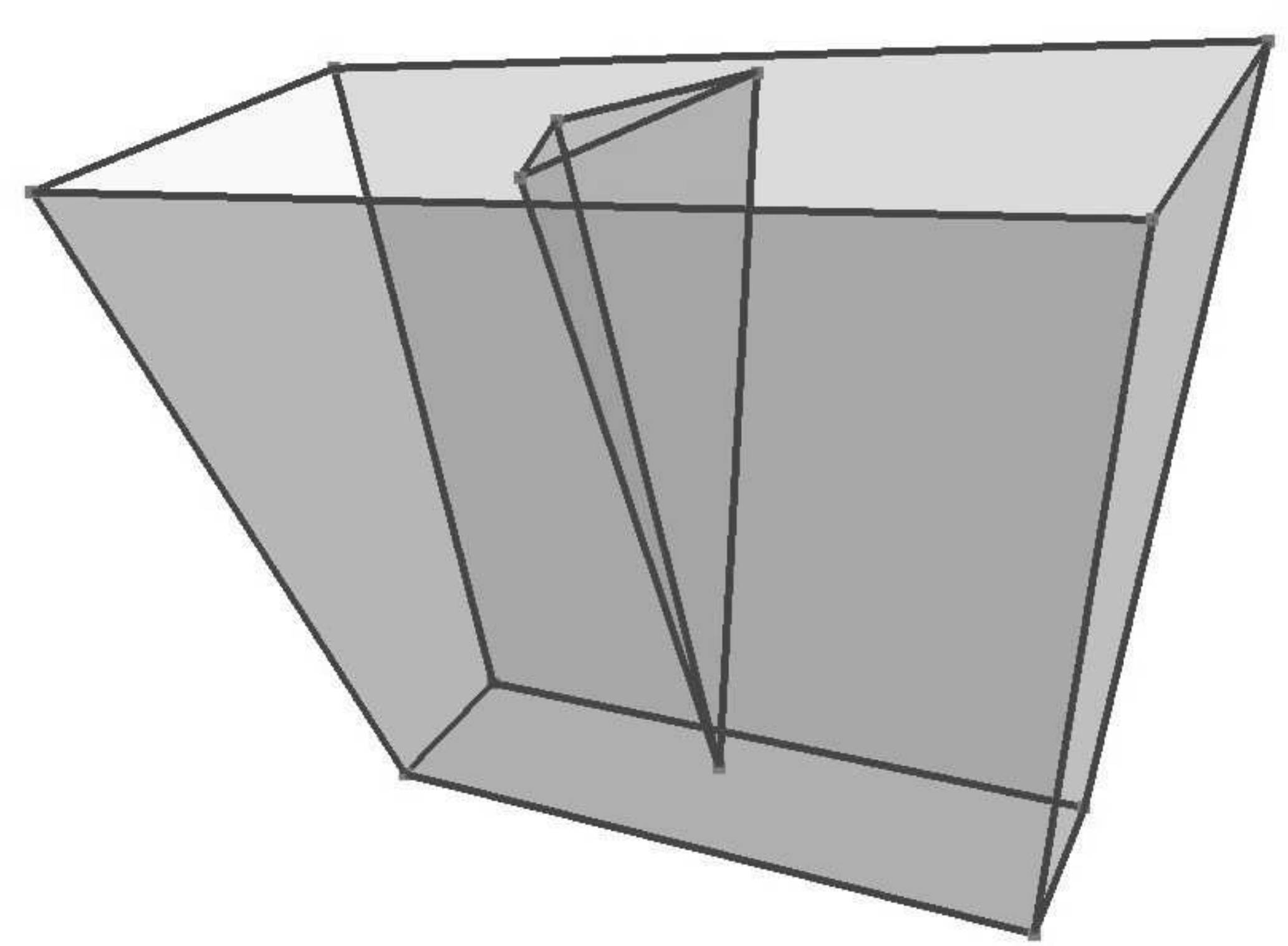} &
         \includegraphics[width=0.90in]{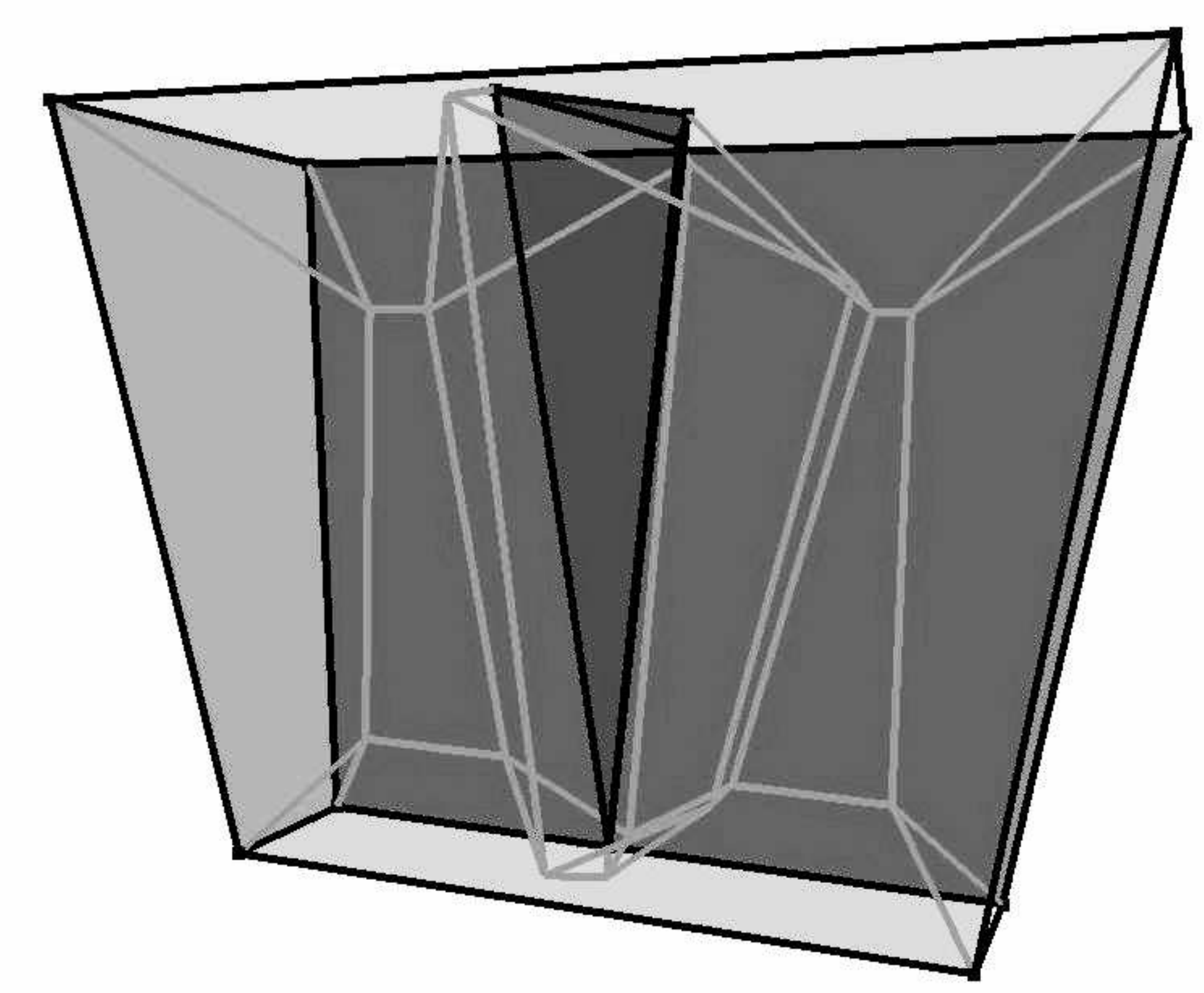} & &
         \includegraphics[width=0.75in]{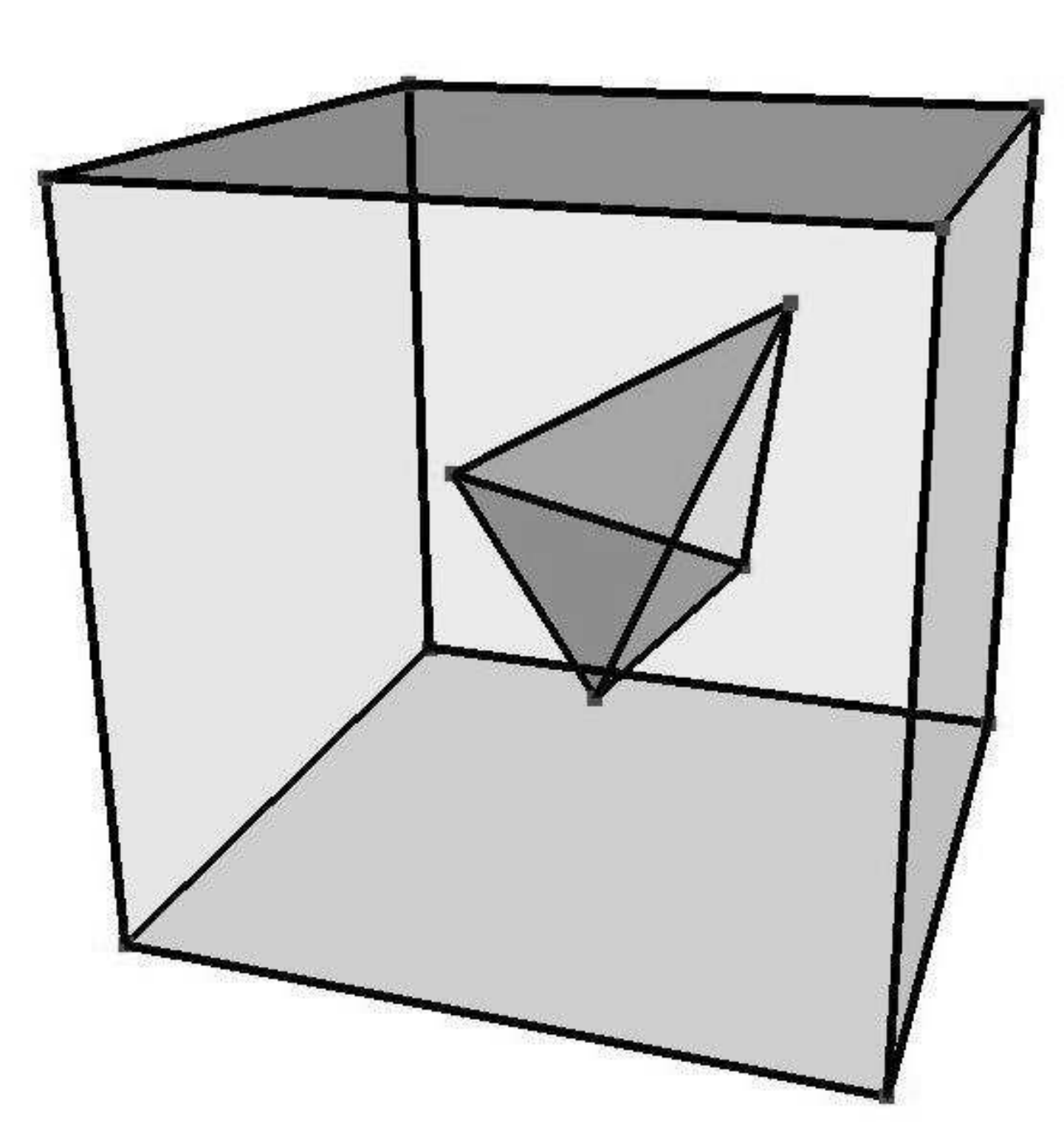} &
         \includegraphics[width=0.75in]{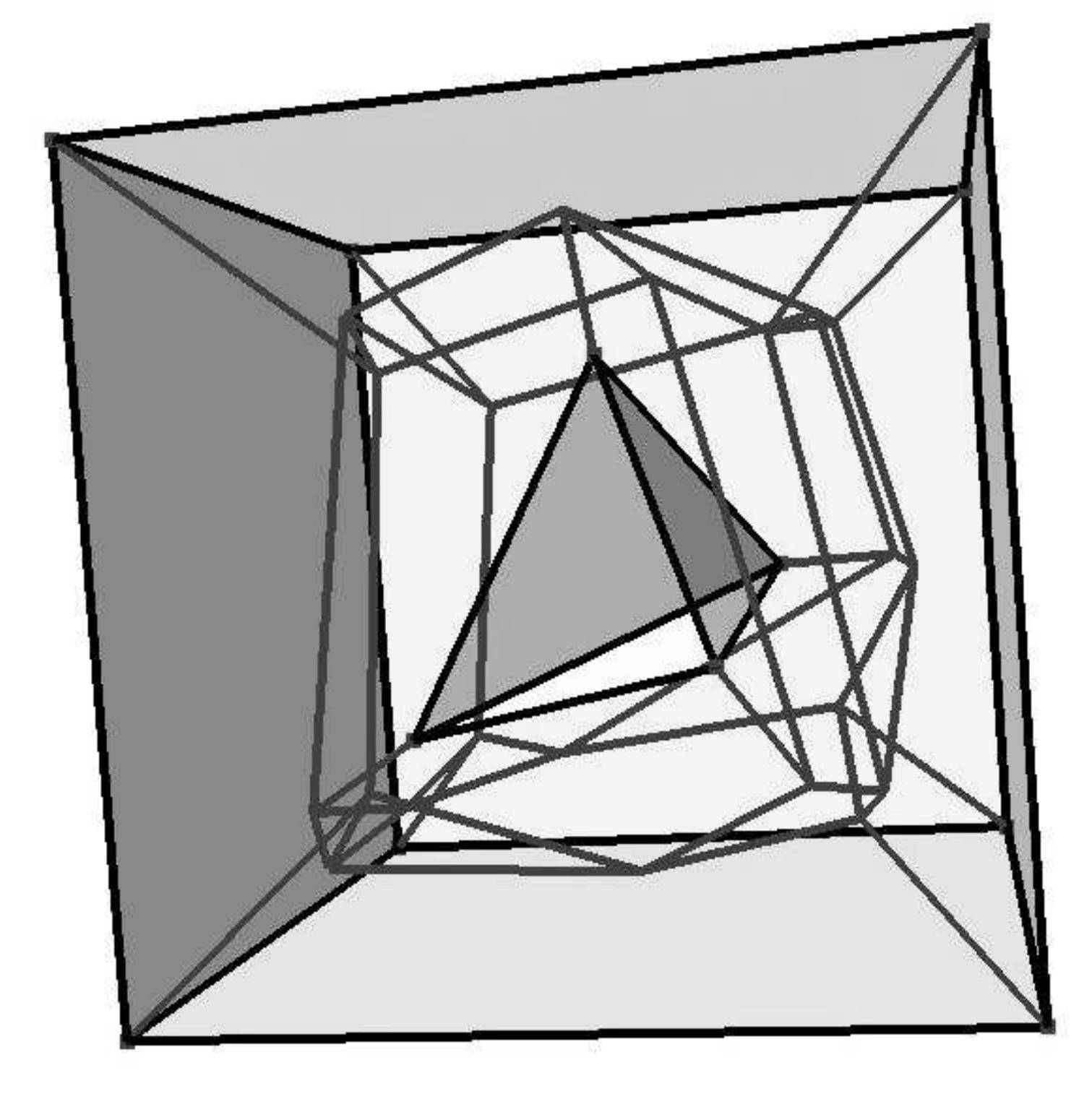} \\
      \multicolumn{2}{c}{(d)} & & \multicolumn{2}{c}{(e)} & &
         \multicolumn{2}{c}{(f)} \medskip \\
      \includegraphics[width=0.95in]{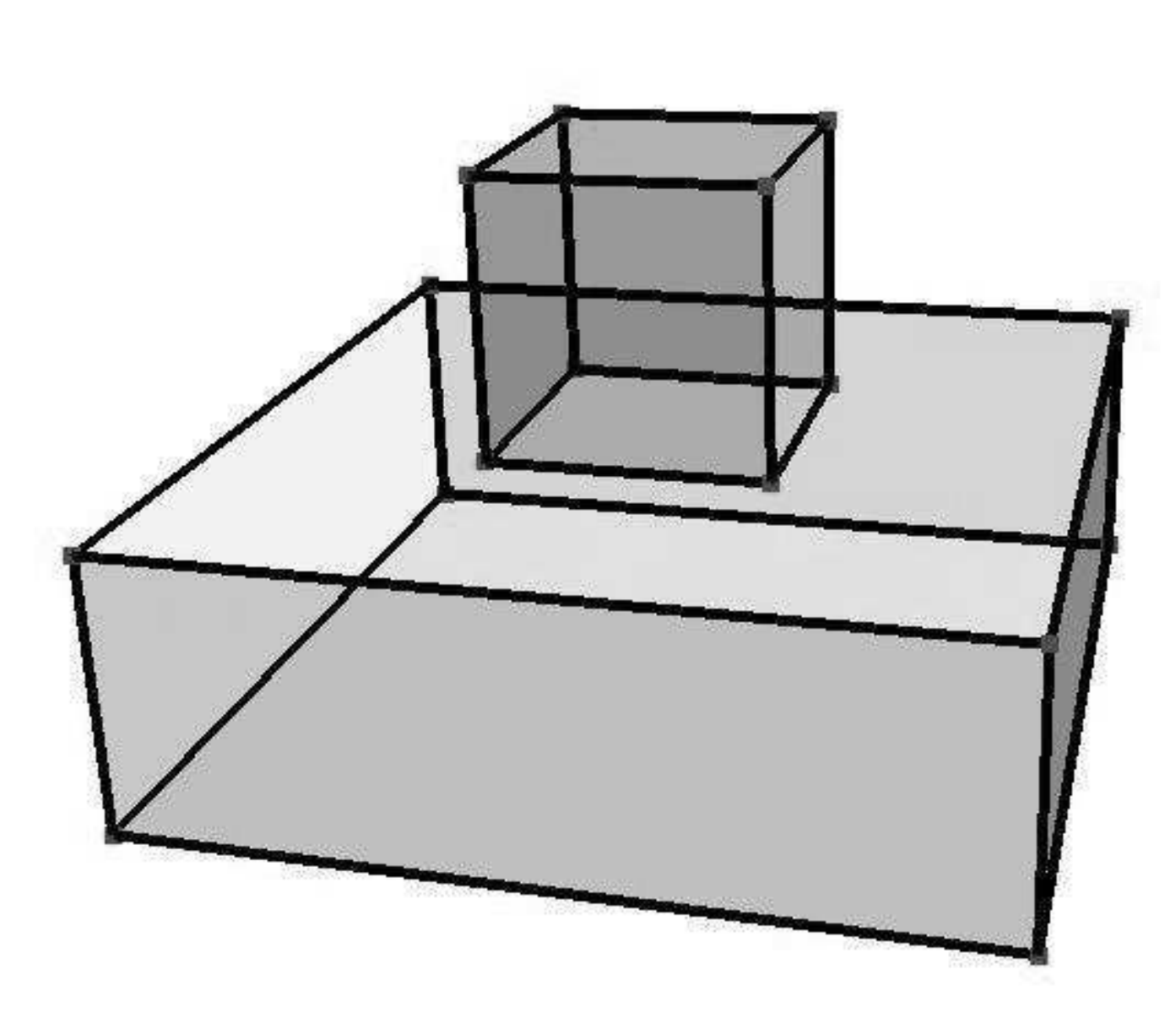} &
         \includegraphics[width=1.00in]{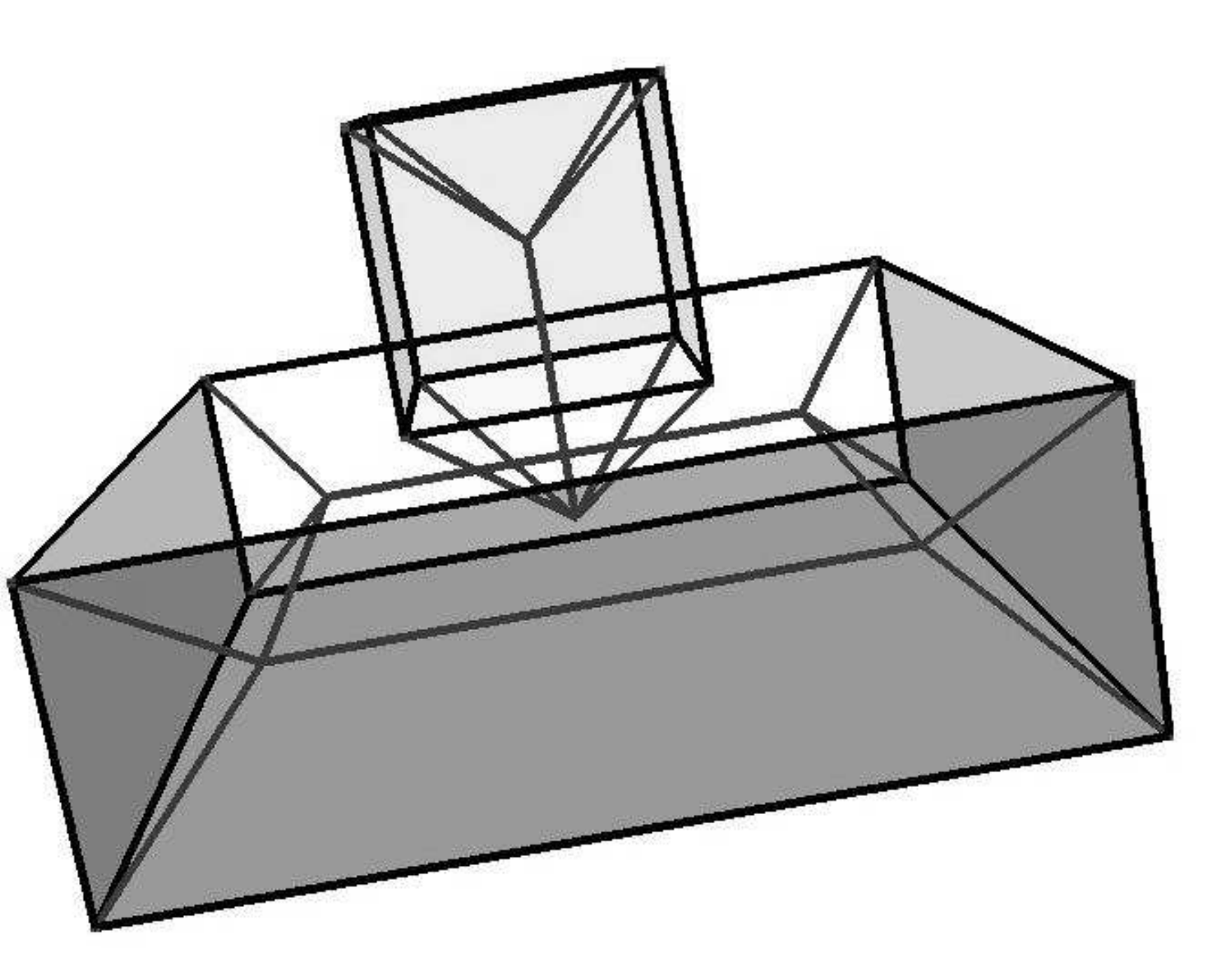} & &
         \includegraphics[width=0.75in]{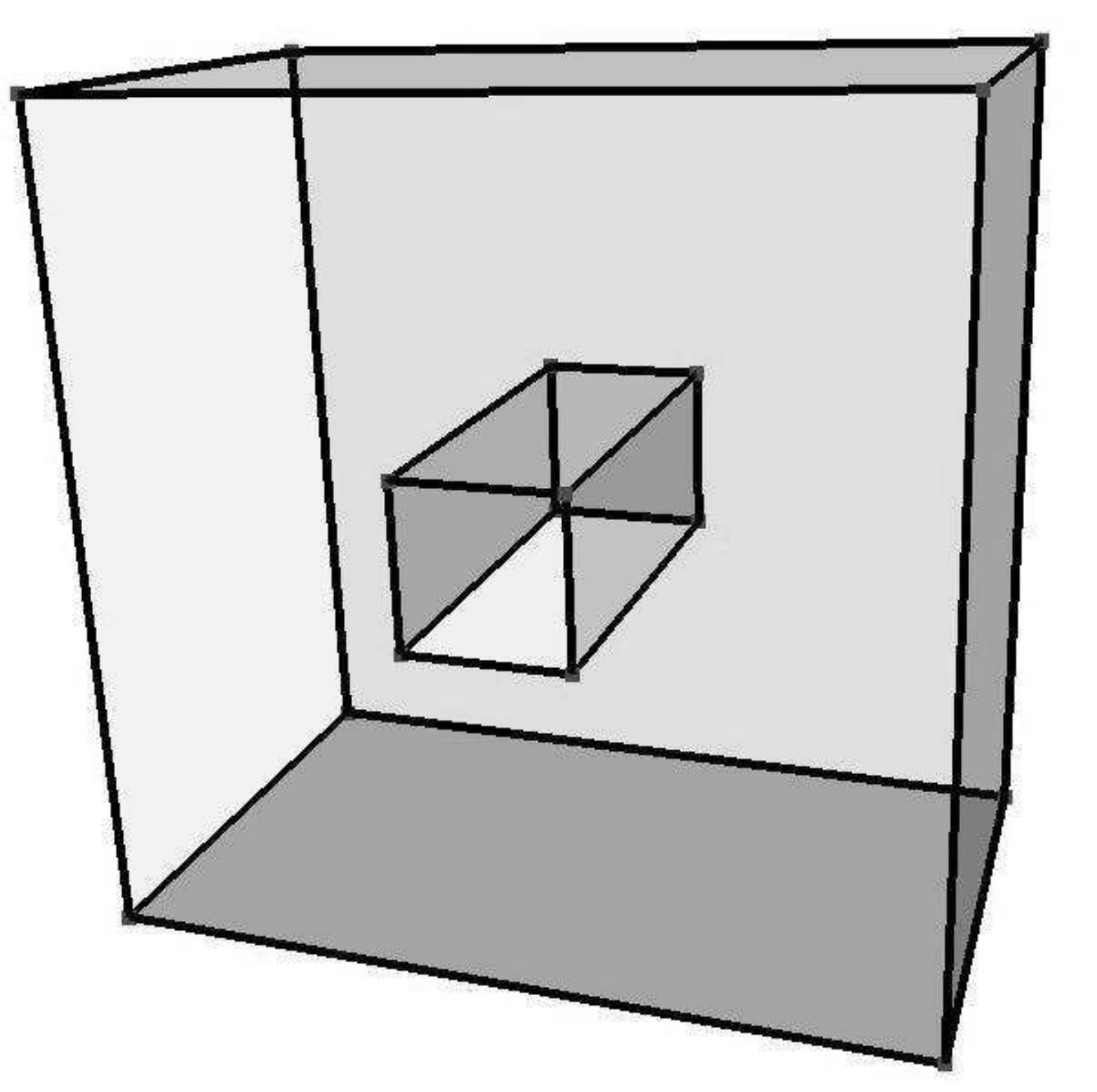} &
         \includegraphics[width=0.75in]{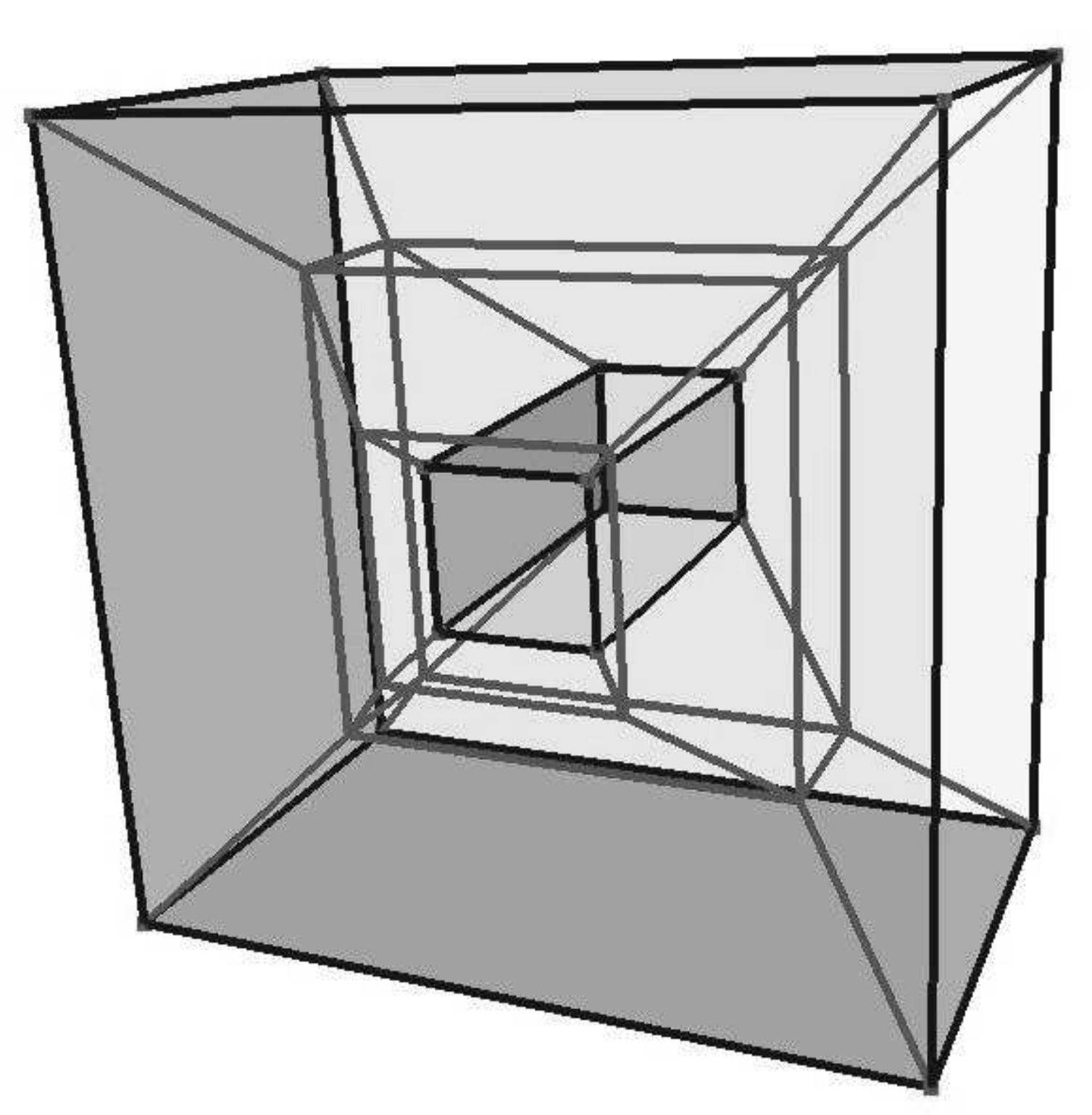} & &
         \includegraphics[width=0.90in]{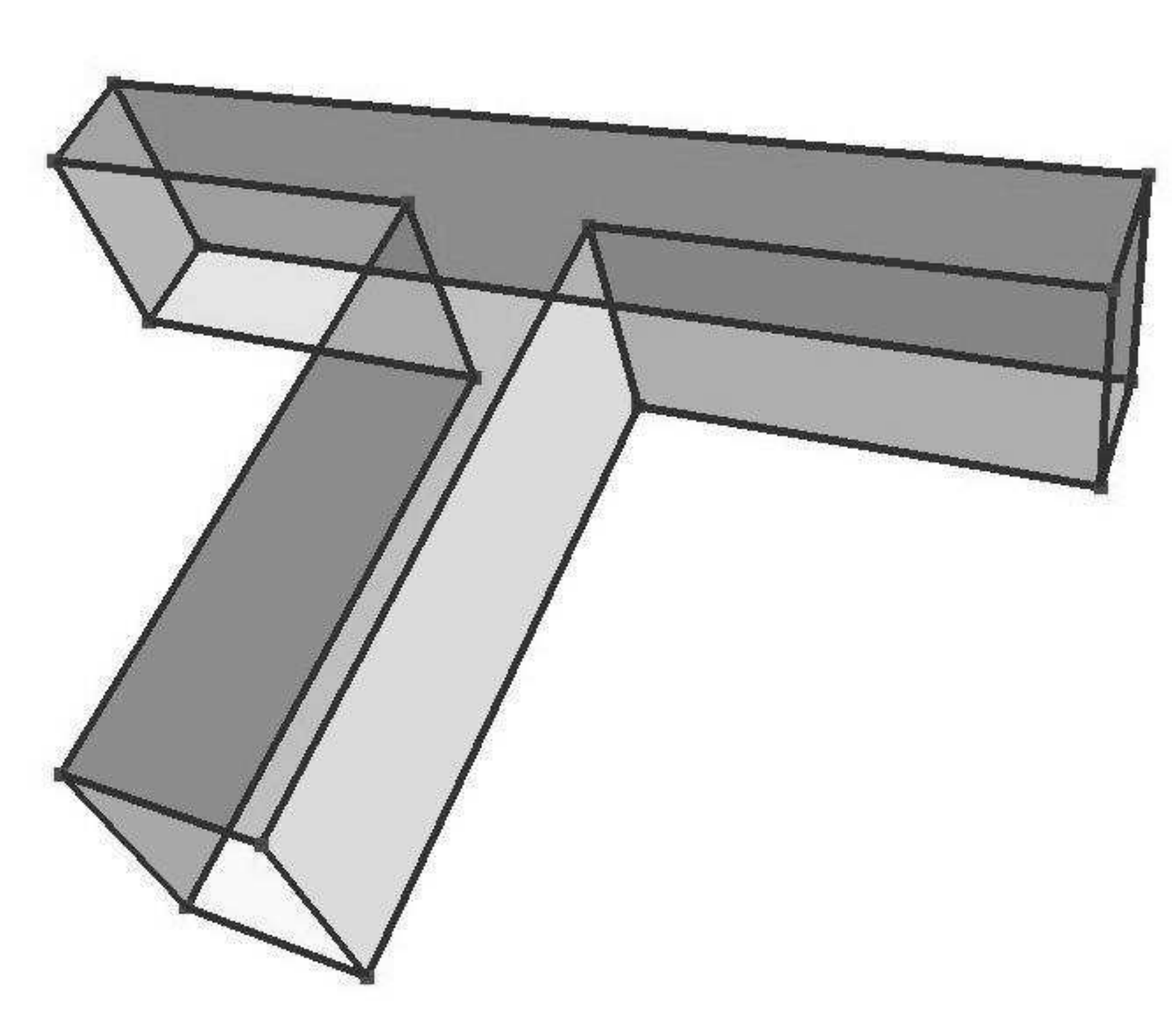} &
         \includegraphics[width=0.95in]{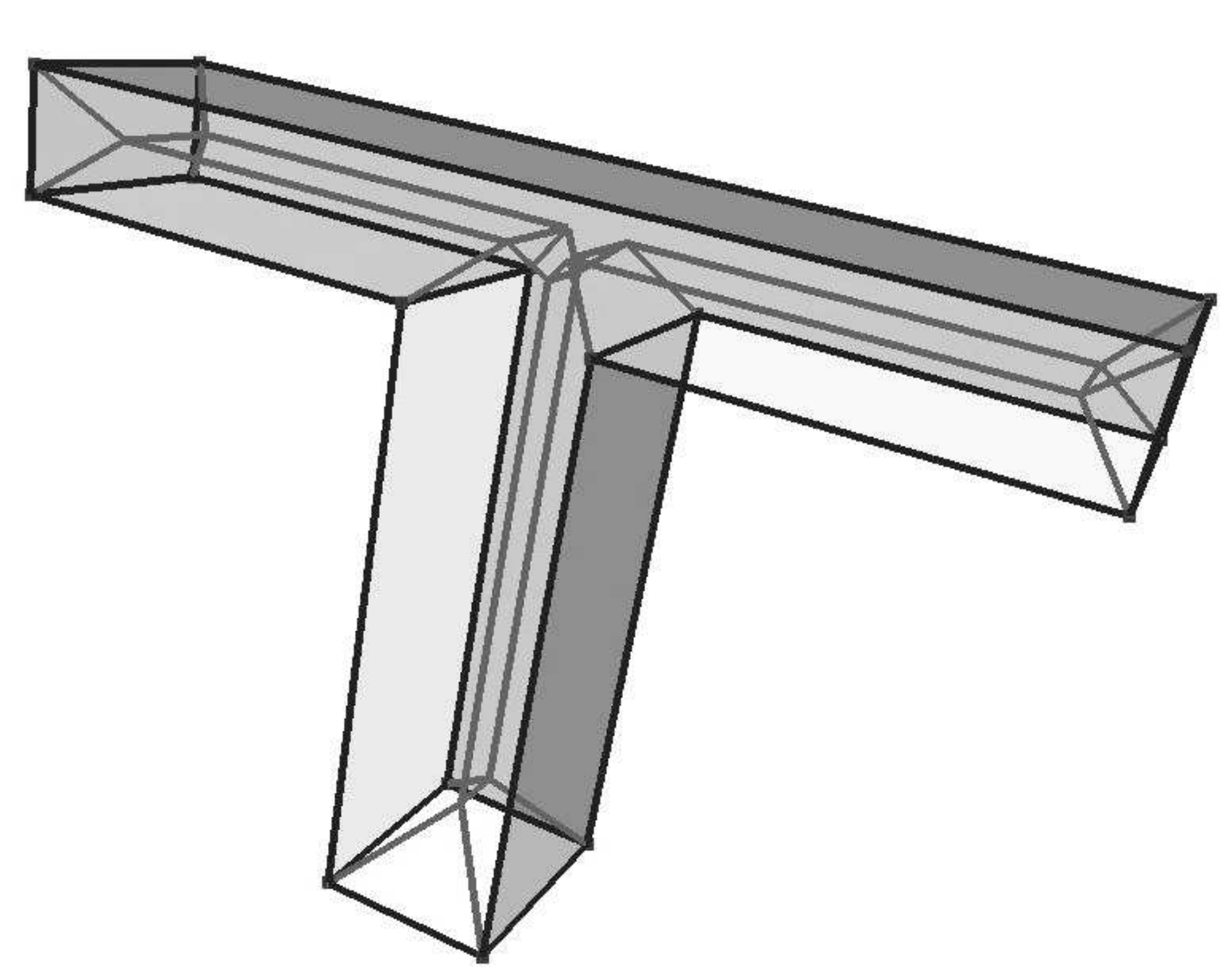} \\
      \multicolumn{2}{c}{(g)} & & \multicolumn{2}{c}{(h)} & &
         \multicolumn{2}{c}{(i)} 
   \end{tabular} \medskip \\
   \small
   \begin{tabular}{|c|ccl|rrcc|r|}
      \hline
          & \multicolumn{3}{c|}{Object} & \multicolumn{4}{c|}{Skeleton} &
      \multicolumn{1}{c|}{Time} \\
      Object & Vertices & Edges & Facets & Vertices & Edges & Faces & Cells &
         \multicolumn{1}{c|}{(Sec.)} \\
      \hline
      \multicolumn{9}{|c|}{General Objects} \\
      \hline
      (a) & 12 & 20~ & 10     &  8~~~ &  24~~ & 25 & 10 & 0.312 \\
      (b) & 20 & 30~ & 12     & 25~~~ &  60~~ & 46 & 12 & 0.719 \\
      (c) & 28 & 42~ & 16     & 45~~~ & 104~~ & 74 & 16 & 0.567 \\
      (d) & 20 & 30~ & 12     & 16~~~ &  42~~ & 37 & 12 & 0.188 \\
      (e) & 20 & 18~ &  9(+1) & 15~~~ &  45~~ & 56 &  9 & 0.250 \\
      (f) & 12 & 18~ & 10     & 21~~~ &  48~~ & 37 & 10 & 0.484 \\
      \hline
      \multicolumn{9}{|c|}{Polycubes} \\
      \hline
      (g) & 16 & 24~ & 11     &  6~~~ &  21~~ & 25 & 11 & 0.177 \\
      (h) & 16 & 24~ & 11     & 12~~~ &  36~~ & 33 & 11 & 0.146 \\
      (i) & 16 & 24~ & 10     & 12~~~ &  32~~ & 29 & 10 & 0.172 \\
      \hline
   \end{tabular} \smallskip \\
   (j) Statistics and running times
   \caption{Sample objects.}
   \label{F-examples}
\end{figure}

\section{Experimental Results}

We have implemented the algorithm for computing the straight skeleton of a
general polyhedron in Visual C++ .NET2005,
and experimented with the software
on a 3GHz Athlon 64 processor PC with 1GB of RAM.  We used the CGAL library
to perform basic geometric operations\ifFullversion, such as plane
intersection, with the embedded rational exact number type GMPQ\fi.
The source code consists of about 6,500 lines of code.
Figure~\ref{F-examples}\ifTitlepage, given in the appendix,\fi{} shows 
the straight skeletons of a few simple objects, and the performance of
our implementation.
(Note that object~(e) contains one hole polygon in addition to the~9
facets.)

\ifTitlepage
   \newpage
\fi
\bibliographystyle{abbrv}
\bibliography{voxskel,edemaine,geom}

\end{document}